\documentclass[11pt]{article}
\usepackage{mydef2col}
\usepackage{kantlipsum,widetext}
\usepackage[T1]{fontenc}
\usepackage[autostyle=true]{csquotes}
\usepackage{pgfplots}
\usepackage{todonotes}

\usepackage[ntheorem]{empheq}
\allowdisplaybreaks[4]
\usepackage{autobreak}

\graphicspath{{./Graphs/}}
\usepackage{tikz}
\usetikzlibrary{calc,decorations.markings}
\usetikzlibrary{arrows,patterns,positioning}

\usepackage[round]{natbib}
\def\calB{{\cal B}}

\def\baru{{\bar{u}}}

\newcommand{\matr}[1]{\mathbf{#1}}

\newcommand{\iu}{\mathrm{i}\mkern1mu}

\newcommand{\LL}{\mathcal{L}}

\DeclareMathOperator*{\res}{Res}

\title{Semi-analytic pricing of double barrier options with time-dependent barriers and rebates at hit}
\def\thetitle1{Semi-analytic pricing of double barrier options...}

\author{
\authorstyle{
Andrey Itkin{}
\textsuperscript{1}
and Dmitry Muravey
\textsuperscript{2}
}
\newline\newline
\textsuperscript{1}
\institution{Tandon School of Engineering, New York University, 1 Metro Tech Center, 10th floor, Brooklyn NY 11201, USA} \\
\textsuperscript{2}
\institution{Moscow State University, Moscow, Russia}
}

\date{\today}
\begin{document}

\maketitle

\lettrineabstract{We continue a series of papers devoted to construction of semi-analytic solutions for barrier options. These options are written on underlying following some simple one-factor diffusion model, but all the parameters of the model as well as the barriers are time-dependent. We managed to show that these solutions are systematically more efficient for pricing and calibration than, eg., the corresponding finite-difference solvers. In this paper we extend this technique to pricing double barrier options and present two approaches to solving it: the General Integral transform method and the Heat Potential method. Our results confirm that for double barrier options these semi-analytic techniques are also more efficient than the traditional numerical methods used to solve this type of problems.}

%%%%%%%%%%%%%%%%%%%%%%%%%%%%%%%%%%%%%%%%%%%%%%%%%%%%%%%%%%%%
\vspace{0.5in}

\section*{Introduction}

Classical problems of financial mathematics recently got new attention due to several factors. Among them one could mention:

\begin{itemize}
\item Very small or even negative interest rates observed at the market, and also forced by the Federal Reserve for achieving its macroeconomic goals, see, eg., \citep{ItkinLiptonMuravey} and reference therein. Therefore. financial models that allow negative rates recently redrew much attention.

\item  Negative oil prices due to the COVID-19 pandemic and the following economic recession, see \citep{Bouchouev2020, Cesa2020}.

\item Another consequence of the COVID-19  was a huge shift to electronic trading since major options exchanges temporarily closed their floors, and brokers and market makers were adjusting to working from home. That raised the need for real-time tools for fast calculating the option prices and Greeks, \citep{Brogan2020}.

\end{itemize}

Those and some other aspects forced the financial society to critically reassess  even simple classical one-factor models of mathematical finance, and reanimate some of them, for instance the Ornstein-Uhlenbeck (OU) process, that traditionally have been referred to as defective/ill-posed or problematic. In \citep{Doff2020} it is advocated that risk managers could even use Black-Scholes to help drive strategy. Therefore, nowadays, for instance, fast pricing of barrier options even for those simple models could be of an increasing importance. That is what this paper is devoted to as applied to double barrier options.

In what follows we consider these options written on the underlying which temporal dynamics is driven by a simple one-factor diffusion process but with time-dependent coefficients. Also, both barriers are assumed to be time dependent. Finally, when the underlying process hits any of the barriers, the Call option holder gets a rebate-at-hit (different for the upper and lower barriers), and they are also time-dependent. It is important that in this paper we consider only the underlying dynamics whose option pricing problem by using the Feynman-Kac theorem and also some transformations could be reduced to the heat equation. Nevertheless, to the best of our knowledge,  even with this simplification a closed-form solution of this problem is yet unknown.

However, we have to mention \citep{Mijatovic2010}, where a similar problem was solved by using a probabilistic approach to obtain a decomposition of the barrier option price into the corresponding European option price minus the barrier premium for a wide class of payoff functions, barrier functions and linear diffusions (i.e., the drift is constant and the local volatility is a function  of the underlying only). For this setting it is shown in \citep{Mijatovic2010} that the barrier premium can be expressed as a sum of integrals along the barriers of the option's delta at the barriers, and that those deltas solve a system of Volterra integral equations of the second kind. This is similar to the idea of the generalized integral transform (GIT) method that we use in this paper, while our setting is more general. Indeed, we allow any diffusion model with time-dependent coefficients and time-dependent barriers and rebates at hit subject to the condition that the pricing partial differential equation (PDE) can be reduced to the heat equation (or, as shown in \citep{CarrItkinMuravey2020} to the Bessel equation). It can also be checked that the pricing PDE in \citep{Mijatovic2010} by a simple change of the spatial variable can be transformed to the heat equation.

Our approach advocated in this paper further extends the technique we elaborated in a series of papers which dealt with a similar problem for single barrier options. In \citep{CarrItkin2020} we developed semi-analytic solutions for the barrier (perhaps, time-dependent) and American options where the underlying stock is driven by a time-dependent OU process with a lognormal drift. This model is equivalent to the familiar Hull-White model in Fixed Income that was separately considered in \citep{ItkinMuravey2020}. In all cases the solution was obtained by using the method of heat potentials (HP) and the GIT method. While the HP method is well-known in mathematical physics and engineering, \citep{TS1963, Friedman1964, kartashov2001}, it is less known as applied to finance. The first use of this method in finance is due to \citep{Lipton2002} for pricing path-dependent options with curvilinear barriers, and more recently in \citep{LiptonKu2018, LiptonPrado2020} (also see references therein).

The GIT method is also known in physics, \citep{kartashov1999, kartashov2001}, but was unknown in finance until the first use  in \citep{CarrItkin2020}. It is important, that it solves the problems where the underlying is defined at the domain $S \in [0,y(t)]$ with $S$ being the stock price, and $y(t)$ being the time-dependent barrier, however, for other domains the solution was unknown even in physics. Then in \citep{ItkinMuravey2020} the GIT solution for the first time was constructed for the domain $S \in [y(t), \infty)$.

Latter this technique was elaborated also for the CIR and CEV models, \citep{CarrItkinMuravey2020}, and the Black-Karasinski model, \citep{ItkinLiptonMuravey}. In particular, in \citep{CarrItkinMuravey2020} the HP method was further generalized to be capable to solving not just the heat but also the Bessel equations, and was called the Bessel potential (BP) method. In \citep{ItkinLiptonMuravey} the PDE is also of a special kind. It is a flavor of the time-dependent Schr\"{o}dinger equation with the unsteady Morse potential (this can be obtained by the change of variables $ x \to -x$ and $\tau \to -\iu \tau$, $\iu = \sqrt{-1}$).

To make it rigorous, in this context a semi-analytic solution means that given a model with the time-dependent drift and  volatility functions, and also with the time-dependent barriers, we obtain the barrier option price in the explicit (analytic) form as an integral in the time $t$. However, this integral contains yet unknown function $\Psi(t)$ which solves some Volterra equation of the second kind which also obtained in our papers. Therefore, we think that "semi-analytic" is an appropriate term. Also, in some situations $\Psi(t)$ can be found analytically, see eg., \citep{CarrItkin2020,ItkinMuravey2020}.

In addition to the explicit analytic representation of the solution, another advantage of this approach is computational speed and accuracy. As this is demonstrated  in the above cited papers, our method is more efficient than both the backward and forward finite difference (FD) methods while providing better accuracy and stability. To briefly explain this, let us mention that the FD method we used (and this is pretty standard) provides accuracy $O(h^2)$ in space  and $O(\tau^2)$ in time, where $h,\tau$ are the corresponding grid steps. Since in our method the solution is represented as a time integral, it can be computed with higher accuracy in time (eg., by using high order quadratures) , while the dependence on the space coordinate $x$ is explicit. Contrary, increasing the accuracy for the FD method is not easy (i.e., it significantly increases the complexity of the method, e.g., see \citep{ItkinBook}). Then the total accuracy is determined by the accuracy of solving the Volterra equation which is also determined by the order of quadratures used to compute the integral in this equation. For instance, using Gaussian quadratures allows small number of nodes and also high accuracy, in more detail see \citep{ItkinMuravey2020,CarrItkinMuravey2020}.

Also, as mentioned in \citep{CarrItkinMuravey2020}, another advantage of our approach is computation of option Greeks. Since the option prices in both the HP and GIT methods are represented in closed form via integrals, the explicit dependence of prices on the model parameters is available and transparent. Therefore, explicit representations of the option Greeks can be obtained by a simple differentiation under the integrals. This means that the values of Greeks can be calculated simultaneously with the prices almost with no increase in time. This is because differentiation under the integrals slightly changes the integrands, and these changes could be represented as changes in weights of the quadrature scheme used to numerically compute the integrals. Since the major computational time has to be spent for computation of densities which contain special functions, they can be saved during the calculation of the prices, and then reused for computation of Greeks.

One can be curious why we need two methods - the HP and GIT, if they are used to solve the same problem and demonstrate the same performance. The answer is kind of elegant. As shown in \citep{CarrItkinMuravey2020}, the GIT method produces very accurate results at high strikes and maturities (i.e. where the option price is relatively small) in contrast to the HP method. This can be verified by looking at the exponents under the GIT solution integral which are {\it proportional} to the time $\tau$. Contrary, when the price is higher (short maturities, low strikes) the GIT method is slightly less accurate than the HP method, as the exponents in the HP  solution integral are {\it inversely proportional} to $\tau$.  Thus, both methods are complementary.

This situation is well investigated for the heat equation with constant coefficients. There exist two representation of the
solution: one - obtained by using the method of images, and the other one - by the Fourier series. Despite
both solutions are equal as the infinite series, their convergence properties are different, \citep{Lipton2002}.

Going back to the problem considered in this paper, we skip the explicit formulation of the model. Instead we define a wide class of models where pricing double barrier options can be translated to solving the heat equation with time-dependent boundaries (barriers) and time-dependent boundary conditions (rebates-at-hit). Note, that the problems considered in the above cited paper - pricing barrier and American options in the time-dependent OU process, pricing barrier options in the Hull-White model, etc., also fit to this class as this is shown in the corresponding papers. Then we construct the solution by using both the GIT and the HP methods. The latter was already shortly presented in \citep{ItkinMuravey2020}, but for the homogeneous boundary conditions. Also, here we present full derivation of the explicit value of the solution spatial gradient $u_x$ at the lower $x=y(\tau)$ and upper $x=z(\tau)$ boundaries. This derivation differs from that in \citep{LiptonKu2018} (and is closer in sense to \citep{TS1963}), but provides a similar result. Also, all the results obtained in this paper are new.

The novelty of the paper is as follows. First, we construct a semi-analytical solution of the heat equation with two arbitrary moving boundaries and arbitrary time-dependent boundary conditions at these boundaries. To the best of authors' knowledge this problem was not solved yet.

Second, various financial problems, where efficient pricing of double barrier options with rebates at hit is subject of investigation, can be reduced to this setting. As we have mentioned it already in above, they include time-dependent Hull-White and OU models, the time-dependent Black-Scholes model, etc., \citep{CarrItkin2020, ItkinMuravey2020, ItkinLiptonMuravey}. Also, for the CIR and CEV models, where the pricing problem is reduced to solving the Bessel PDE with time-dependent boundaries, the latter can also be solved in a similar manner, \citep{CarrItkinMuravey2020}. Also, local volatility models with $\sigma = \sigma(x)$ can be also treated under this setting.

Third, consider a general one-factor model
\begin{equation} \label{OU}
d S_t = \mu(t,S) dt + \sigma(t,S)dW_t, \qquad S_t(t=0) = S_0.
\end{equation}
\noindent where $t > 0$ is the time, $S_t$ is the spot price, $\mu(t,S)$ is the drift, $\sigma(t,S)$ is the volatility of the process, $W_t$ is  the standard Brownian motion under the risk-neutral measure. This model can be solved as follows. Let us split the  domain of the definition of $S$ into $N$ intervals, and at every interval $i=1,\ldots,N$ approximate the drift by a linear function of $S$, i.e. $\mu_i(t,S) = a_i(t) + b_i(t)S$, and the volatility - by a quadratic function $\sigma_i(t,S) = c_i(t) + d_i(t) S + e_i(t)S^2$. Then it can be shown that at every interval the corresponding pricing PDE can be transformed to the heat equation with time-dependent boundaries and the boundary conditions. Using continuity of the solution and its gradient at every sub-boundary, this problem can be solved semi-analytically in a similar fashion. In physics this approach is called the method of multilayer heat equation, see, eg., a nice survey in \citep{Dias2014}. In more detail the development of this method as applied to finance will be published elsewhere. Thus, solving (semi-analytically) the heat equation with time-dependent moving boundaries and the boundary conditions is a key element of such a method. Having this method in hands, pricing double barrier options for any financial model of the type \eqref{OU} can be done semi-analytically.

The rest of the paper is organized as follows. Section~\ref{statement} describes the double barrier pricing problem
for the time-dependent barriers and rebates at hit and shows that it can be reduced to solving inhomogeneous PDE with homogeneous boundary conditions. Section~\ref{sec:GIT} describes in detail the solution of this problem by using the GIT method. We provide two alternative integral representations of the solution - one via the Jacobi theta functions, and the other one - using the Poisson summation formula. Despite these solutions are equal in a sense of infinite series, their convergence properties are different. A system of the Volterra equations for the gradient of the solution at both boundaries is obtained for both representations. Section~\ref{sec:HP} provides the same development but using the HP method. The final section concludes.

\section{Statement of the problem} \label{statement}

Let us consider a one-factor diffusion model in \eqref{OU} By using a standard argument, to price options written on $S_t$ as an underlying, one can apply the Feynman-Kac theorem to obtain the following partial differential equation (PDE) for, eg., the European Call option price
\begin{equation} \label{PDE}
\fp{C}{t} + \dfrac{1}{2}\sigma^2(t,S) \sop{C}{S} +  \mu(t,S) S \fp{C}{S} = r(t) C.
\end{equation}
Here in case of Equities we treat $S_t$ as the stock price, then $r(t)$ is the deterministic  interest rate. If $S_t$ is the stochastic interest rate, then $r(t)$ in the RHS of \eqref{PDE} should be replaced with $S$.

The \eqref{PDE} should be solved subject to the terminal condition at the option maturity $t=T$
\begin{equation} \label{tc0}
C(T,S) = (S-K)^+,
\end{equation}
\noindent where $K$ is the option strike, and some boundary conditions. Below in this paper we are concentrated only on double barrier options with moving barriers: the lower barrier at $S = L(t)$ and the upper barrier at $S = H(t) > L(t)$, so $S \in [L(t),H(t)]$.

Our main assumption in this paper is that the PDE in \eqref{PDE} by a series of transformations of the dependent variable $C(S,t) \mapsto U(x,\tau)$ and independent variables $S \mapsto x(t,S), \ t \mapsto \tau(t,S)$  can be reduced to the heat equation
\begin{equation} \label{Heat}
\fp{U}{\tau} = \sop{U}{x},
\end{equation}
\noindent which should be solved at the new domain $x \in [y(\tau),z(\tau)], \ \tau \in [0, \tau(0,S_0)]$, subject to the terminal condition
\begin{equation} \label{tc}
U(0,x) = U_0(x),
\end{equation}
\noindent and the boundary conditions
\begin{align} \label{bc}
U(\tau, y(\tau)) &= f^-(\tau), \qquad U(\tau, z(\tau)) = f^+(\tau).
\end{align}
Here $f^\pm(\tau), y(\tau), z(\tau)$ are some continuous functions of time $\tau$. From the financial point of view the problem in \eqref{Heat}, \eqref{tc}, \eqref{bc} (the $\calB$ problem) could be viewed as a pricing problem for double barrier options with the moving lower $y(\tau)$ and upper $z(\tau)$ barriers and the rebates $f^\pm(\tau)$ paid at hit, i.e. when the underlying process hits either the lower or the upper  barrier.

Note, that many well-known financial models fit this framework. For instance, the time dependent OU process used in \citep{CarrItkin2020} to model barrier and American options  is such an example. Also, the time-dependent Hull-White model considered in \citep{ItkinMuravey2020} for pricing barrier options is another example. The number of models that fit this framework could be significantly expanded if one transforms the original PDE in \eqref{PDE} to its multilayer version. This approach is discussed it detail in \citep{ItkinLiptonMuraveyMulti} and will be reported elsewhere.

Below we present solution of the $\calB$ problem by using two analytic methods - the GIT and HP methods. As mentioned in Introduction, the methods are complementary in a sense that despite both solutions are equal, their convergence properties are different. In particular, the GIT method is more accurate at high strikes and maturities while the HP method - at low strikes and maturities.

It is worth mentioning that the $\calB$ problem is with inhomogeneous boundary conditions, hence from the very beginning it is useful to transform it to a similar problem but with homogeneous boundary conditions. This could be done by the change of variables
\begin{align} \label{homTrans}
u(\tau,x) &= U(\tau,x) - A(\tau) - B(\tau) x, \\
A(\tau) &= - \frac{f^+(\tau) y(\tau) - f^-(\tau) z(\tau)}{z(\tau) - y(\tau)}, \qquad
B(\tau) = \frac{f^+(\tau) - f^-(\tau) }{z(\tau) - y(\tau)}, \nonumber
\end{align}
\noindent which transforms the PDE in \eqref{PDE} to the inhomogeneous PDE but with the homogeneous boundary conditions
\begin{align} \label{wEq}
\fp{u}{\tau} &= \sop{u}{x} + g(\tau,x), \\
g(\tau,x) & \equiv - A'(\tau) - B'(\tau) x, \quad (\tau, x) \in \mathbb{R}_+ \times [y(\tau), z(\tau)], \nonumber  \\
u(0,x) &= U_0(x) - A(0) - B(0) x \equiv u_0(x), \qquad u(\tau, y(\tau)) = u(\tau, z(\tau)) = 0. \nonumber
\end{align}

\section{Solution by the GIT method} \label{sec:GIT}

In this section we solve the problem in \eqref{wEq} by using the GIT method, see \citep{kartashov1999, CarrItkin2020, ItkinMuravey2020,ItkinLiptonMuravey} and references therein. However, as mentioned in \citep{kartashov2001}, an analytic solution for the domain with two moving boundaries is yet unknown. Therefore, our solution presented in this Section is new, and it extends the results of \citep{CarrItkin2020} obtained for the domain $[0, y(\tau)]$.

In \citep{CarrItkin2020} the authors used the GIT proposed in \citep{kartashov1999} which is a map
$u(\tau,x) \mapsto \bar{u}(\tau,p)$ of the form
\begin{equation} \label{trCI}
 \bar{u}(\tau,p) = \int_0^{y(\tau)} u(\tau,x) \sinh(x \sqrt{p}) dx,
 \end{equation}
 \noindent where $p = a + \iu \omega$ is a complex number with $\Re(p) \ge \beta > 0$, and $- \frac{\pi}{4} <
 \arg\left(\sqrt{p}\right) < \frac{\pi}{4}$. Here we proceed with a similar idea by introducing the transform
 \begin{equation} \label{eq:GIT_def}
\baru(\tau, p) = \int_{y(\tau)}^{z(\tau)} u(\tau, x) \sinh \left( p[x - y(\tau)] \right)dx.
\end{equation}
With a special choice of the lower boundary $y(\tau) = 0$ this transform replicates that one in \eqref{trCI} subject to the point that here we use the spectral parameter $p$ instead of $\sqrt{p}$ as in \citep{CarrItkin2020}.
%Therefore, here
%$- \frac{\pi}{4} <  \arg\left(p\right) < \frac{\pi}{4}$.

Since the kernel of \eqref{eq:GIT_def} is time-dependent it doesn't make much sense to apply this transform directly to the inhomogeneous heat equation in \eqref{wEq}. Therefore, we represent the image $\baru$ as a difference  of two other images
\begin{equation} \label{eq:GIT_def2}
\baru = \frac{1}{2}(\baru_+ - \baru_-), \qquad
\baru_\pm (\tau, p) = \int_{y(\tau)}^{z(\tau)} u(\tau, x) e^{\pm p [x - y(\tau)]}dx.
\end{equation}

To determine $\baru(\tau,p)$ let us multiply both parts of the first line in \eqref{wEq} by $e^{\pm p [x - y(\tau)]}$ and integrate on $x$. These yield
\begin{align} \label{heatTr}
\int_{y(\tau)}^{z(\tau)} & \fp{u(\tau, x)}{\tau} e^{\pm p [x - y(\tau)]}dx  = \fp{\baru_\pm(\tau, p)}{\tau} - u(\tau, z(\tau)) e^{\pm p z(\tau)} z'(\tau) + u(\tau, y(\tau)) e^{ \pm p y(\tau)}y'(\tau) \\
&\pm   p y'(\tau) \int_{y(\tau)}^{z(\tau)} u(\tau, x) e^{\pm p [x - y(\tau)]}dx = \fp{\baru_\pm}{\tau} \pm p y'(\tau) \baru_\pm, \nonumber \\
\int_{y(\tau)}^{z(\tau)} & \sop{u(\tau, x)}{x} e^{\pm p [x - y(\tau)]}dx = \left[\Phi(\tau)- B(\tau)\right] e^{\pm p [z(\tau) - y(\tau)]} + \left[\Psi(\tau) + B(\tau) \right]  + p^2 \baru_\pm(\tau, p), \nonumber \\
\bar{g}_{\pm}(\tau,p) &\equiv \int_{y(\tau)}^{z(\tau)} g(\tau,x) e^{\pm p[x - y(\tau)]} d x = \frac{B'(\tau)}{p^2}
\left( e^{\pm p [z(\tau) - y(\tau)]} - 1 \right) \nonumber \\
&\pm \frac{1}{p} \left[A'(\tau) \left(1 - e^{\pm p [z(\tau) - y(\tau)]}\right) + B'(\tau) \left(y(\tau) -z(\tau) e^{\pm p [z(\tau) - y(\tau)]}\right) \right]. \nonumber
 \end{align}
\noindent where terms proportional to $u(\tau,y(\tau)$ and $u(\tau,z(\tau)$ vanish due to the boundary conditions in \eqref{wEq}, and by definition
\begin{align} \label{PsiPhi_def}
\Psi(\tau) = -\fp{U(\tau, x)}{x} \Bigg|_{x = y(\tau)} &\quad
\Phi(\tau) = \fp{U(\tau, x)}{x} \Bigg|_{x = z(\tau)}.
%-\fp{u(\tau, x)}{x} \Bigg|_{x = y(\tau)} = \Psi(\tau) + B(\tau), &\quad
% \fp{u(\tau, x)}{x} \Bigg|_{x = z(\tau)} = \Phi(\tau) - B(\tau)
\end{align}

Collecting terms in \eqref{heatTr} yields two initial value problems, one for the function $\baru_+$ and the other one - for $\baru_-$
\begin{align} \label{baru_pm_ODE}
\fp{\baru_\pm(\tau,p)}{\tau} &+ \baru_\pm \left[ \pm p y'(\tau) - p^2\right]
=  \left[\Psi(\tau) + B(\tau)\right] + \left[\Phi(\tau) - B(\tau)\right] e^{\pm p[z(\tau) - y(\tau)]} + \bar{g}_{\pm}(\tau,p), \\
\baru_\pm(0, p) &= \int_{y(0)}^{z(0)} u(0, x) e^{\pm p[x - y(0)]} dx. \nonumber
\end{align}

Each problem in \eqref{baru_pm_ODE} (for the plus and minus signs) can be solved explicitly
\begin{align} \label{baru_pm_sols}
\baru_\pm(\tau, p) &= e^{p^2 \tau} \int_{y(0)}^{z(0)} u(0, x) e^{\pm p[x - y(\tau)]} dx \\
&+\int_0^\tau e^{p^2 (\tau - s)} \left[ \left[\Phi(s) - B(s) \right] e^{\pm p[z(s) - y(\tau)]} + \left(\Psi(s) +B(s) + \bar{g}_\pm(s,p) \right)
e^{\pm p[y(s) - y(\tau)]}\right] ds. \nonumber
\end{align}

Note that the last term in the second integral in \eqref{baru_pm_sols} can be re-written in a more convenient form
\begin{align*}
\bar{g}_{\pm}(s,p) & e^{\pm p \left[ y(s) - y(\tau)\right]} = \frac{B'(s)}{p^2}
\left( e^{\pm p [z(s) - y(s)]} - 1 \right)  e^{\pm p \left[ y(s) - y(\tau)\right]} \\
&\pm \frac{1}{p} \left[A'(s) \left(1 - e^{\pm p [z(s) - y(s)]}\right)  e^{\pm p \left[ y(s) - y(\tau)\right]} + B'(s) \left(y(s) -z(s) e^{\pm p [z(s) - y(s)]}\right)   e^{\pm p \left[ y(s) - y(\tau)\right]}\right] \\
&= \frac{B'(s)}{p^2} \left( e^{\pm p [z(s) - y(\tau)]} - e^{\pm p [y(s) - y(\tau)]} \right) \\
&\pm \frac{1}{p} \left[A'(s) \left(e^{\pm p [y(s) - y(\tau)]} - e^{\pm p [z(s) - y(\tau)]}\right) + B'(s) \left(y(s)e^{\pm p [y(s) - y(\tau)]}   -z(s) e^{\pm p [z(s) - y(\tau)]}\right) \right].
\end{align*}

The explicit representation for $\baru$ then follows from its definition in \eqref{eq:GIT_def2}
\begin{align} \label{baru_sol}
\baru(\tau, p) &= e^{p^2 \tau} \int_{y(0)}^{z(0)} u(0, x) \sinh\left(p[x - y(\tau)] \right) dx \\
&+ \int_0^\tau e^{p^2 (\tau - s)} \left[ \left[\Phi(s) -B(s) \right]  \sinh\left(p[z(s) - y(\tau)]\right) + \left[\Psi(s) + B(s) \right] \sinh(p[y(s) - y(\tau)]) + h(s,p) \right] ds, \nonumber \\
h(s,p) &= \frac{B'(s)}{p^2} \left[\sinh(p [z(s)-y(\tau)]) - \sinh(p[y(s)-y(\tau)]) \right] \nonumber \\
&+  \frac{1}{p}\left[ \left(A'(s) + B'(s) y(s)\right) \cosh(p [y(s)-y(\tau)])
- \left(A'(s) + B'(s) z(s)\right) \cosh(p [z(s) - y(\tau)])  \right]. \nonumber
\end{align}

\subsection{The inverse transform}

General theory of the heat equation tells us that the solution at the space domain $a< x < b, \ a,b \in \Re - const$, can be represented as Fourier series of the form,  \citep{Polyanin2002})
\[
u(\tau, x) = \sum_{n = 1}^{\infty}\alpha_n e^{- \frac{\pi^2 n ^2}{(b- a)^2}\tau }
\sin \left( \frac{\pi n (x - a)}{b- a}\right)
\]
Therefore, by analogy let us look for the inverse transform of $\baru$ (which actually is the solution $u(\tau, x)$ of \eqref{wEq}) to be a generalized Fourier transform of the form (\cite{CarrItkin2020})
\begin{equation} \label{series_u_def}
u(\tau, x) = \sum_{n =0}^{\infty} A_n(\tau) \sin\left( \pi n \frac{x- y(\tau)}{z(\tau) - y(\tau)}\right),
\end{equation}
\noindent where $A_n(\tau)$ are some yet unknown Fourier coefficients (weights). Applying the direct transform in \eqref{eq:GIT_def} to the series in \eqref{series_u_def} yields
\begin{equation} \label{baru_series}
\baru(\tau, x) = \int_{y(\tau)}^{z(\tau)} \sum_{n =1}^{\infty} A_n(\tau) \sin\left( \pi n \frac{x- y(\tau)}{z(\tau) - y(\tau)}\right) \sinh\left(p [x - y(\tau)]\right) dx.
\end{equation}

Using the identity
\begin{equation}
\int_{y}^{z} \sin \left(\pi n \frac{x - y }{z-y}\right) \sinh \left(p [x - y]\right) dx
= (-1)^{n + 1} \frac{\pi n (z - y) \sinh \left(p [z - y]\right)  }{n^2 \pi^2 + p^2 (z- y )^2},
\end{equation}
\noindent we obtain another representation for $\baru$
\begin{equation} \label{baru_sol_series}
\baru(\tau, x) = \frac{1}{l(\tau)} \sum_{n = 1}^{\infty}\frac{ (-1)^{n + 1} \pi n A_n(\tau) \sinh \left( p l(\tau)\right)}{ \left[p + \iu n \pi/l(\tau)\right] \left[p - \iu n \pi/l(\tau) \right]  }, \qquad
l(\tau) = z(\tau) - y(\tau).
\end{equation}

Combining \eqref{baru_sol_series} and \eqref{baru_sol} yields the equation for $A_n(\tau)$
\begin{align} \label{ID_for_residuals}
\frac{1}{l(\tau)} \sum_{n = 1}^{\infty} & \frac{ (-1)^{n + 1} \pi n A_n(\tau)}{ \left[p + \iu n \pi/l(\tau)\right] \left[p - \iu n \pi/l(\tau) \right] } =  \frac{1}{\sinh \left( p\, l(\tau)\right)} \Bigg \{e^{p^2 \tau} \int_{y(0)}^{z(0)} u(0, x) \sinh\left(p[x - y(\tau)] \right) dx \\
+&\int_0^\tau e^{p^2 (\tau - s)} \left[ \left[\Phi(s) - B(s)\right]  \sinh\left(p[z(s) - y(\tau)]\right) + \left[\Psi(s) + B(s) \right]  \sinh\left(p[y(s) - y(\tau)]\right)
+ h(s,p) \right] ds \Bigg\}. \nonumber
\end{align}

The LHS and RHS of \eqref{ID_for_residuals} as the functions of $p$ are analytic in the whole complex plane domain except the poles
\begin{equation}
p_{k}^{\pm} = \pm \iu \pi k/l(\tau), \quad k=1,2,\ldots,
\end{equation}
\noindent because $h(s,p)$ is regular and well-behaved at $p \to 0$.  Also, as this is easy to check, these poles are common for the LHS and RHS of \eqref{ID_for_residuals}. For what follows we need the following residues
\begin{equation} \label{residual_values}
\res_{p = p_k^\pm} \sum_{n = 1}^{\infty} \frac{1}{ \left[p + \iu n \pi/l(\tau)\right] \left[p - \iu n \pi/l(\tau) \right] } = \pm \frac{l(\tau)}{2\iu \pi k}, \qquad \qquad
\res_{p = p_k^\pm}  \frac{1}{\sinh\left(p\, l(\tau)\right)} = \frac{(-1)^k}{l(\tau)}.
\end{equation}

The Fourier coefficients $A_k(\tau)$  can now be found from \eqref{ID_for_residuals} by applying contour integration on $p$ to both sides. We integrate using the contours  $L_k^+, \ k=1,2,\ldots$, where the integration contours look like it is depicted in Fig.~\ref{contour}. Thus, we have
\begin{align} \label{cInt}
\frac{1}{l(\tau)}\int\displaylimits_{L_k^+} &\sum_{n = 1}^{\infty}\frac{ (-1)^{n + 1} \pi n A_n(\tau)}{ \left[p + \iu n \pi/l(\tau)\right] \left[p - \iu n \pi/l(\tau) \right] } dp = \int\displaylimits_{L_k^+} \frac{1}{\sinh \left( p l(\tau)\right)} \Bigg\{ e^{p^2 \tau} \int_{y(0)}^{z(0)} u(0, x) \sinh\left(p[x - y(\tau)] \right) dx \nonumber \\
&+ \int_0^\tau e^{p^2 (\tau - s)} \left[ \Phi(s)  \sinh\left(p[z(s) - y(\tau)]\right) + \Psi(s)  \sinh\left(p[y(s) - y(\tau)]\right) + h(s,p) \right] ds \Bigg\} dp.
\end{align}

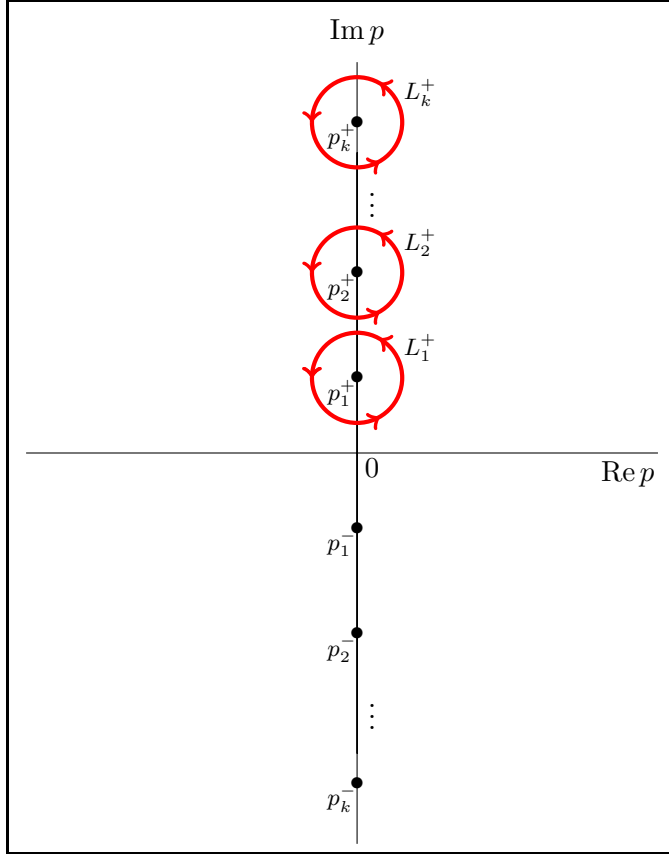
\begin{figure}[!hb]
\hspace{-0.5in}
\captionsetup{width=0.8\linewidth}	
\begin{center}
\caption{Contours of integration of \eqref{ID_for_residuals} in the complex plane $p \in \mathbb{C}$ with poles at $p_1^\pm, p_2^\pm, \ldots$.}
\label{contour}
\fbox{
			\begin{tikzpicture}
			\begin{scope}[scale=4.]
			% Configurable parameters
			\def\bigradius{1}
			\def\axisradius{1}
			\def\gammaradius{0}
			\def\omegaradius{0}
			\def\lradius{0.15}
			\def\pone{0.25}
			\def\ptwo{0.6}
			\def\pthird{1.1}
			\def\dr{1 / 3 /3}
			% Axes
			\draw (-1.1*\axisradius, 0) -- (1.*\axisradius,0)
			(\omegaradius, -1.3*\axisradius) -- (\omegaradius, 1.3*\axisradius)
			(\gammaradius, -\bigradius) -- (\gammaradius, \bigradius);
			
			\draw[red, ultra thick, decoration={ markings,
				 %% ANTI_CLOCKWISE
				 mark=at position \dr * 0 with {\arrow{>}}
				,mark=at position \dr * 1 with {\arrow{>}}
				,mark=at position \dr * 2 with {\arrow{>}}
				,mark=at position \dr * 3 with {\arrow{>}}
				,mark=at position \dr * 4 with {\arrow{>}}
				,mark=at position \dr * 5 with {\arrow{>}}
				,mark=at position \dr * 6 with {\arrow{>}}
				,mark=at position \dr * 7 with {\arrow{>}}
				,mark=at position \dr * 8 with {\arrow{>}}
				,mark=at position \dr * 9 with {\arrow{>}}
				%% CLOCK WISE
%				,mark=at position \dr * 10 with {\arrow{<}}
%				,mark=at position \dr * 11 with {\arrow{<}}
%				,mark=at position \dr * 12 with {\arrow{<}}
%				,mark=at position \dr * 13 with {\arrow{<}}
%				,mark=at position \dr * 14 with {\arrow{<}}
%				,mark=at position \dr * 15 with {\arrow{<}}
%				,mark=at position \dr * 16 with {\arrow{<}}
%				,mark=at position \dr * 17 with {\arrow{<}}
%				,mark=at position \dr * 18 with {\arrow{<}}
			},
			postaction={decorate}]
			let in
    			(- \lradius,\pone ) arc(-180:180:\lradius) 			
    			(- \lradius,\ptwo ) arc(-180:180:\lradius) 			
    			(- \lradius,\pthird ) arc(-180:180:\lradius) ;			
	%		(- \lradius,-\pone ) arc(-180:180:\lradius) 			
	%		(- \lradius,-\ptwo ) arc(-180:180:\lradius) 			
	%		(- \lradius,-\pthird) arc(-180:180:\lradius) ;			
			%% The labels			;			
			\node at (0.9*\axisradius,-0.07){$\operatorname{Re} p$};
			\node at (0,1.4*\axisradius) {$\operatorname{Im} p$};

			\node at (0.05,-0.05){$0$};
			\node at (0, \pone){$\bullet$};
			\node at (0, -\pone){$\bullet$};
			\node at (0, \ptwo){$\bullet$};
			\node at (0, -\ptwo){$\bullet$};
			\node at (0, \pthird){$\bullet$};
			\node at (0, -\pthird){$\bullet$};
            \node at (0.05, -\pthird + 0.25){$\vdots$};
            \node at (0.05, \pthird - 0.25){$\vdots$};

			%%%%%% 1st pole %%%%%%%			
			\node at (-0.05, \pone -0.05){{\footnotesize $p_1^+$}};
			\node at (-0.05, -\pone -0.05){{\footnotesize $p_1^-$}};
			\node at (0.06 +\lradius, \pone +0.1){{\footnotesize $L_1^+$}};
%			\node at (0.06 +\lradius, -\pone +0.1){{\footnotesize $L_1^-$}};
			%%%%%% 2nd pole %%%%%%%	
			\node at (-0.05, \ptwo -0.05){{\footnotesize $p_2^+$}};
			\node at (-0.05, -\ptwo -0.05){{\footnotesize $p_2^-$}};
			\node at (0.06 +\lradius, \ptwo +0.1){{\footnotesize $L_2^+$}};
%			\node at (0.06 +\lradius, -\ptwo +0.1){{\footnotesize $L_2^-$}};
			%%%%%% n-th pole %%%%%%%	
			\node at (-0.05, \pthird -0.05){{\footnotesize $p_k^+$}};
    		\node at (-0.05, -\pthird -0.05){{\footnotesize $p_k^-$}};
			\node at (0.06 +\lradius, \pthird +0.1){{\footnotesize $L_k^+$}};
%			\node at (0.06 +\lradius, -\pthird +0.1){{\footnotesize $L_k^-$}};

			\end{scope}
			\end{tikzpicture}
		}
	\end{center}
\end{figure}

By the Cauchy's residue  theorem each integral in \eqref{cInt} is equal to the sum of the corresponding residues  that can be computed with the help of \eqref{residual_values}. This yields the following formula for $A_k(\tau)$
\begin{equation} \label{ResIntPlus}
A_k(\tau) =  \frac{2}{\iu l(\tau)} \baru\left(\tau, \iu \frac{\pi k}{l(\tau)} \right).
\end{equation}

With allowance for \eqref{baru_sol} this can be finally represented as
\begin{align} \label{A_k_formula}
A_k(\tau) &= \frac{2}{l(\tau)} \Bigg\{
e^{-\frac{\pi^2 k^2}{l^2(\tau)} \tau} \int_{y(0)}^{z(0)} u(0, x) \sin\left(\frac{\pi k}{l(\tau)}[x - y(\tau)] \right) dx \\
&+ \int_0^\tau e^{-\frac{\pi^2 k^2}{l^2(\tau)} (\tau - s)} \bigg[\left[\Phi(s) - B(s)\right] \sin\left(  \frac{\pi k}{l(\tau)} [z(s) - y(\tau)]\right) \nonumber \\
&+ \left[\Psi(s) +B(s)\right]  \sin\left( \frac{\pi k}{l(\tau)}[y(s) - y(\tau)]\right) + h_1(k, s,\tau) \bigg] ds  \Bigg\}, \nonumber
\end{align}
\noindent with
\begin{align} \label{h1}
&h_1(k, s,\tau) = -\frac{B'(s) l^2(\tau)}{\pi^2 k^2} \left[\sin\left( \frac{\pi k }{l(\tau)} [z(s)-y(\tau)] \right) - \sin\left( \frac{\pi k }{l(\tau)}[y(s)-y(\tau)]\right) \right] \\
&-\frac{l(\tau)}{\pi k}\Bigg[ \left(A'(s) + B'(s) y(s)\right) \cos\left( \frac{\pi k }{l(\tau)} [y(s)-y(\tau)]\right) - \left(A'(s) + B'(s) z(s)\right) \cos\left( \frac{\pi k }{l(\tau)} [z(s) - y(\tau)]\right)  \Bigg]. \nonumber
\end{align}
Keeping in mind that
\begin{equation*}
A(\tau) + B(\tau) y(\tau) = f^-(\tau) \qquad A(\tau) + B(\tau) z(\tau) = f^+(\tau)
\end{equation*}
\noindent we re-arrange \eqref{h1} as
\begin{align} \label{h12}
&h_1(k, s,\tau) = -\frac{B'(s) l^2(\tau)}{\pi^2 k^2} \left[\sin\left( \frac{\pi k }{l(\tau)} [z(s)-y(\tau)] \right) - \sin\left( \frac{\pi k }{l(\tau)}[y(s)-y(\tau)]\right) \right] \\
&- \frac{l(\tau)}{\pi k}\Bigg[ \left((f^-)'(s) - B(s)y'(s)\right) \cos\left( \frac{\pi k }{l(\tau)} [y(s)-y(\tau)]\right) - \left((f^+)'(s) - B(s) z'(s)\right) \cos\left( \frac{\pi k }{l(\tau)} [z(s) - y(\tau)]\right)  \Bigg]. \nonumber
\end{align}
Substituting this result into \eqref{series_u_def}, we obtain the solution $u(\tau, x)$ of the problem \eqref{wEq}
\begin{align} \label{u_final}
u(\tau, &x) = \frac{2}{l(\tau)}  \sum_{n= 1}^{\infty} \sin\left(\pi n  \frac{x - y(\tau)}{l(\tau)}\right) \Bigg\{
e^{-\frac{\pi^2 n^2}{l^2(\tau)} \tau} \int_{y(0)}^{z(0)} u(0, \xi) \sin\left(\frac{\pi n}{l(\tau)}[\xi - y(\tau)] \right) d\xi \\
&+ \int_0^\tau e^{-\frac{\pi^2 n^2}{l^2(\tau)} (\tau - s)} \bigg[\left[\Phi(s) - B(s)\right]  \sin\left(  \frac{\pi n}{l(\tau)} [z(s) - y(\tau)]\right) + \left[\Psi(s) + B(s)\right]  \sin\left( \frac{\pi n}{l(\tau)}[y(s) - y(\tau)]\right)  \bigg] ds  \nonumber \\
& + \int_0^\tau e^{-\frac{\pi^2 n^2}{l^2(\tau)} (\tau - s)} h_1(n,s,\tau) ds \Bigg\}. \nonumber
\end{align}

This expression can be further simplified, see Appendix~\ref{App1}.  Returning back  to the original variable $U(\tau, x)$ yields the final representation
\begin{align} \label{U_final}
U(\tau, x) &= \frac{2}{l(\tau)}  \sum_{n= 1}^{\infty} \sin\left(\pi n  \frac{x - y(\tau)}{l(\tau)}\right) \Bigg\{ e^{-\frac{\pi^2 n^2}{l^2(\tau)} \tau} \int_{y(0)}^{z(0)} U(0, \xi) \sin\left(\frac{\pi n}{l(\tau)}[\xi - y(\tau)] \right) d\xi \\
&+ \int_0^\tau e^{-\frac{\pi^2 n^2}{l^2(\tau)} (\tau - s)} \Big[ \Phi(s)  \sin\left(  \frac{\pi n}{l(\tau)} [z(s) - y(\tau)]\right) + \Psi(s)  \sin\left( \frac{\pi n}{l(\tau)}[y(s) - y(\tau)]\right) \nonumber \\
&+ \beta(\tau, s, n) \Big] ds  \Bigg\} + F(\tau,x). \nonumber
\end{align}
\noindent where $\beta(\tau,  s, n)$ and $F(\tau,x)$ are defined in \eqref{alpha_and_beta_def} and \eqref{Fdef_alt}. Also, as can be checked from the definition in \eqref{Fdef_alt} that  at $y(\tau) < x < z(\tau)$ the function $F(\tau,x)$ vanishes, and $F(\tau,y(\tau)) = f^-(\tau), \ F(\tau,z(\tau)) = f^+(\tau)$. Thus, \eqref{U_final} solves the problem in \eqref{Heat} with the initial condition in \eqref{tc} and the boundary conditions in \eqref{bc}.

It is worth mentioning that the exact same formalism can be developed by using another integral transform
\begin{equation*}
\baru(\tau, p) = \int_{y(\tau)}^{z(\tau)} \sinh\left(p[z(\tau) - x]\right) u(\tau, x) dx,
\end{equation*}
\noindent with the result being same as in \eqref{u_final}.

\subsection{Connection to the Jacobi theta function}

As observed in \citep{CarrItkin2020}, the sums in \eqref{u_final} could be expressed via the Jacobi theta functions of the third kind, \citep{mumford1983tata}\footnote{Which is not a surprise since it is known that the Jacobi theta functions is the solution of the heat equation with periodic boundary conditions. As applied to the problem considered in this paper, an example is a double barrier option with zero rebate at hit.}. Using their definition
\begin{equation} \label{theta_def}
\theta_3(z,\omega) = 1 + 2\sum_{n = 1}^{\infty} \omega^{n^2} \cos \left(2 n z \right),
\end{equation}
\noindent and the identities
\begin{align} \label{theta_integrals}
\fp{\theta_3(z,\omega)}{z} = \theta_3'(z,\omega) = -4 \sum_{n = 1} ^{\infty}
n \omega^{n^2} \sin \left( 2 n z\right).
\end{align}
\noindent we obtain from \eqref{u_final}
\begin{alignat}{2} \label{thetaRepr}
&4\sum_{n = 1}^{\infty} e^{-\frac{\pi^2 n^2}{l^2(\tau)} \tau} \sin \left(  \frac{n \pi  (x - y(\tau))}{l(\tau)}\right) \sin \left(  \frac{n \pi  (\xi - y(\tau))}{l(\tau)}\right) &&= \theta_3(\phi_-(x,\xi), \omega_1) -  \theta_3(\phi_+(x,\xi), \omega_1), \\
&4\sum_{n = 1}^{\infty} e^{-\frac{\pi^2 n^2}{l^2(\tau)} (\tau - s)} \sin \left(  \frac{n \pi  (x - y(\tau))}{l(\tau)}\right) \sin \left(  \frac{n \pi  (\xi - y(\tau))}{l(\tau)}\right) &&= \theta_3(\phi_-(x,\xi), \omega_2) -  \theta_3(\phi_+(x,\xi), \omega_2), \nonumber \\
&8\sum_{n = 1}^{\infty} n e^{-\frac{\pi^2 n^2}{l^2(\tau)} (\tau - s)} \sin \left(  \frac{n \pi  (x - y(\tau))}{l(\tau)}\right) \cos \left(  \frac{n \pi  (\xi - y(\tau))}{l(\tau)}\right) &&= -\left( \theta_3'(\phi_-(x,\xi),\omega_2) + \theta_3'(\phi_+(x,\xi), \omega_2) \right). \nonumber
\end{alignat}
\begin{equation*}
\omega_1 = e^{-\frac{\pi^2 \tau}{l^2(\tau)}}, \quad \omega_2 = e^{-\frac{\pi^2 (\tau - s)}{l^2(\tau)}}, \quad
\phi_-(x,\xi) = \frac{\pi (x - \xi)}{2 l(\tau)}, \quad \phi_+(x,\xi) = \frac{\pi (x + \xi - 2y(\tau))}{2 l(\tau)}.
\end{equation*}

With the help of \eqref{thetaRepr} the final formula for $u(\tau,x)$ simplifies
\begin{align}  \label{u_final2}
2 l(\tau) \Big[U(\tau, x) &- F(\tau,x)\Big] = \int_{y(0)}^{z(0)} U(0, \xi) \left[ \theta_3(\phi_-(x,\xi), \omega_1) -  \theta_3(\phi_+(x,\xi), \omega_1) \right]  d\xi \\
&+\int_0^\tau \Bigg\{ \left[\Psi(s)- f^-(s) y'(s)\right] \left[ \theta_3(\phi_-(x,y(s)), \omega_2) -  \theta_3(\phi_+(x,y(s)), \omega_2) \right] \nonumber \\
&+ \left[\Phi(s)+ f^+(s) z'(s)\right] \left[ \theta_3(\phi_-(x,z(s)), \omega_2) -  \theta_3(\phi_+(x,z(s)), \omega_2) \right] \nonumber \\
&+ \frac{1}{2} \Big[ f^+(s) \left[\theta_3'(\phi_-(x,z(s)), \omega_2) + \theta_3'(\phi_+(x,z(s)),\omega_2)\right] \nonumber \\
&- f^-(s)\left[\theta_3'(\phi_-(x,y(s)), \omega_2) + \theta_3'(\phi_+(x,y(s)),\omega_2) \right] \Big]  \Bigg\} ds. \nonumber
\end{align}

Note, that if rebates at hit are not paid, the boundary conditions become homogeneous, and all terms proportional to $f^-(s) = f^+(s) = 0$ in \eqref{u_final2} disappear.

\subsection{Determining  $\Psi(\tau)$ and $\Phi(\tau)$}

Taking the derivative in \eqref{u_final2} with respect to $x$, having in mind that according to \eqref{theta_integrals}
\begin{align}
\fp{\theta_3(\phi_\pm(x,\xi),\omega_2)}{x} &= \frac{\pi}{l(\tau)} \fp{\theta_3(y,\omega_2)}{y}\Bigg|_{y = \phi_\pm(x,\xi)} = \frac{\pi}{l(\tau)}  \theta_3'(\phi_\pm(x,\xi),\omega_2),  \\
\sop{\theta_3(\phi_\pm(x,\xi),\omega_2)}{x} &= \frac{\pi^2}{l^2(\tau)} \sop{\theta_3(y,\omega_2)}{y}\Bigg|_{y = \phi_\pm(x,\xi)} =
\frac{\pi^2}{l^2(\tau)} \theta_3''(\phi_\pm(x,\xi),\omega_2), \nonumber
\end{align}
\noindent  and substituting $x = y(\tau)$ and $x = z(\tau)$, we get a system of Volterra integral equations of the second kind to determine $\Psi(\tau, \Phi(\tau)$
\begin{align}  \label{VolterraTheta}
-\frac{2 l^2(\tau)}{\pi} \Big[\Psi(\tau) &+ F_x(\tau,y(\tau))\Big] = \int_{y(0)}^{z(0)} U(0, \xi) \left[ \theta'_3(\phi_-(y(\tau),\xi), \omega_1) -  \theta'_3(\phi_+(y(\tau),\xi), \omega_1) \right]  d\xi \\
&+\int_0^\tau \Bigg\{ \left[\Psi(s) - f^-(s) y'(s)\right] \left[ \theta'_3(\phi_-(y(\tau),y(s)), \omega_2) -  \theta'_3(\phi_+(y(\tau),y(s)), \omega_2) \right] \nonumber \\
&+ \left[\Phi(s) + f^+(s) z'(s)\right] \left[ \theta'_3(\phi_-(y(\tau),z(s)), \omega_2) -  \theta'_3(\phi_+(y(\tau),z(s)), \omega_2) \right] \nonumber \\
&+ \frac{2 \pi}{l(\tau)} \Big[ f^+(s) \left[\theta_3''(\phi_-(y(\tau), z(s)), \omega_2) + \theta_3''(\phi_+(y(\tau), z(s)),\omega_2)\right] \nonumber \\
&- f^-(s)\left[\theta_3''(\phi_-(y(\tau),y(s)), \omega_2) + \theta_3''(\phi_+(y(\tau), y(s)),\omega_2) \right] \Big]  \Bigg\} ds. \nonumber \\
\frac{2 l^2(\tau)}{\pi} \Big[\Phi(\tau) &+ F_x(\tau,z(\tau))\Big] = \int_{y(0)}^{z(0)} U(0, \xi) \left[ \theta'_3(\phi_-(z(\tau),\xi), \omega_1) -  \theta'_3(\phi_+(z(\tau),\xi), \omega_1) \right]  d\xi \nonumber  \\
&+\int_0^\tau \Bigg\{ \left[\Psi(s) - f^-(s) y'(s)\right] \left[ \theta'_3(\phi_-(z(\tau),y(s)), \omega_2) -  \theta'_3(\phi_+(z(\tau),y(s)), \omega_2) \right] \nonumber \\
&+ \left[\Phi(s) + f^+(s) z'(s)\right] \left[ \theta'_3(\phi_-(z(\tau),z(s)), \omega_2) -  \theta'_3(\phi_+(z(\tau),z(s)), \omega_2) \right] \nonumber \\
&+ \frac{2 \pi}{l(\tau)} \Big[ f^+(s) \left[\theta_3''(\phi_-(z(\tau), z(s)), \omega_2) + \theta_3''(\phi_+(z(\tau), z(s)),\omega_2)\right] \nonumber \\
&- f^-(s)\left[\theta_3''(\phi_-(z(\tau),y(s)), \omega_2) + \theta_3''(\phi_+(z(\tau), y(s)),\omega_2) \right] \Big]  \Bigg\} ds. \nonumber
\end{align}

Also, since the theta function $\theta_3(z,\omega)$ solves the heat equation
\begin{equation*}
\fp{\theta_3(z,\iu t)}{t} =  \frac{1}{4\pi} \sop{\theta_3(z,\iu t)}{z},
\end{equation*}
\noindent the second derivatives with respect to the first argument could be expressed via the first derivatives with respect to the second argument.

However, there exists a problem with the representation in \eqref{VolterraTheta}. Indeed, using the definition of $F(\tau,x)$ in \eqref{Fdef_alt} it can be checked that the derivatives $F_x(\tau,x)$ do not exist at $x=y(\tau)$ and $x=z(\tau)$ as they are proportional to the Dirac Delta $\delta(0)$. Therefore, in the next Section we attack this problem again using an alternative representation of the solution.

\subsection{The Poisson summation formula and alternative representations}
\label{sec:Poisson}
It is known that for the fixed spatial domain $x \in [y(\tau), z(\tau)], \ y(\tau) = 0, \ z(\tau) = const$ there exist two representations of the solution of the heat equation: one - obtained by using the method of images, and the other one - by the Fourier series. Both solutions are equal in a sense of infinite series, but their convergence properties are different, see eg., \citep{Lipton2002}. It turns out that for a curvilinear strip we can also obtain an alternative representation.

The solution $u(\tau,x)$ found in \eqref{u_final} already has the form of the Fourier series. However, applicability of the method of images for the problem \eqref{wEq} is not transparent due to time-dependency of the boundaries. Instead, we can find an alternative representation by using the following property known as the Poisson Summation formula, \citep{PolBremer1950}
\begin{proposition}[Poisson Summation formula] 	\label{prop:poisson}
Let $\hat{h}(\nu)$ be the Fourier transform of the appropriate function $h(x)$
\begin{equation*}
\hat{h}(\nu) = \int_{-\infty} ^{\infty} h(x) e^{-2\pi \iu \nu x} dx.
\end{equation*}
The following identity holds
\begin{equation} \label{poisson_sums}
\sum_{n = -\infty}^{\infty} h(n) = \sum_{k = -\infty}^{\infty} \hat{h}(k).
\end{equation}
\end{proposition}
\begin{proof}
  See \citep{PolBremer1950}.
\end{proof}

Applying \eqref{poisson_sums} to the functions
\begin{alignat*}{2}
h_1(x) &= e^{- \frac{\pi^2 x^2}{2}} \cos \left( \pi x \alpha \right), \qquad
&&\hat{h}_1(\nu) =  \int_{-\infty}^{\infty} e^{- \frac{\pi^2 x^2}{2\beta} - 2\pi \iu \nu x} \cos \left( \pi x \alpha \right) dx, \\
h_2(x) &=  x e^{- \frac{\pi^2 x^2}{2\beta}} \sin \left( \pi x \alpha \right), \qquad
&&\hat{h}_2(\nu) =  \int_{-\infty}^{\infty}  x e^{- \frac{\pi^2 x^2}{2\beta} - 2\pi \iu \nu x} \sin \left( \pi x \alpha \right) dx,
\end{alignat*}
\noindent we obtain the following identities
\begin{align} \label{Summation_formula_main_eq}
\sum_{n = -\infty}^{\infty}  e^{-\frac{\pi^2 n^2 }{2 \beta}} \cos \left( \pi n \alpha \right)
&= \sqrt{\frac{2 \beta} {\pi}} e^{- \frac{\alpha ^2 \beta}{ 2} }  \sum_{n = -\infty}^{\infty} e^{-2n^2 \beta } \cosh \left(2 n \alpha \beta \right) \\
&= \sqrt{\frac{\beta} {2\pi}} \sum_{n = -\infty}^{\infty} \left[ e^{-\frac{ \beta}{2} \left(2n - \alpha \right)^2} + e^{-\frac{ \beta }{2} \left(2n + \alpha \right)^2} \right] = 2\sqrt{\frac{\beta} {2\pi}} \sum_{n = -\infty}^{\infty} e^{-\frac{ \beta }{2} \left(2n + \alpha \right)^2}, \nonumber \\
\sum_{n = -\infty}^{\infty}  \pi n e^{-\frac{\pi^2 n^2 }{2 \beta}} \sin \left( \pi n \alpha \right)
&= \frac{ \beta^{3/2}}{\sqrt{2 \pi }} \sum_{n = -\infty}^{\infty} e^{-\frac{ \beta }{2} \left(2n + \alpha \right)^2} \left[\alpha + 2 n  + (\alpha - 2 n )e^{4 \alpha \beta n} \right] \nonumber \\
&= \frac{ \beta^{3/2}}{\sqrt{2 \pi }} \sum_{n = -\infty}^{\infty} \left[ e^{-\frac{ \beta }{2} \left(2n + \alpha \right)^2} (\alpha + 2 n)  + e^{-\frac{ \beta }{2} \left(2n - \alpha \right)^2} (\alpha - 2 n ) \right] \nonumber \\
&=2 \frac{ \beta^{3/2}}{\sqrt{2 \pi }} \sum_{n = -\infty}^{\infty} e^{-\frac{ \beta }{2} \left(2n + \alpha \right)^2} (\alpha + 2 n). \nonumber
\end{align}

Since each summand in \eqref{U_final} can be represented in the form of the LHS of  \eqref{Summation_formula_main_eq}, by using a simple trigonometric formula for the product of sines, we immediately arrive at another form of $U(\tau, x)$, see Appendix~\ref{App2}
\begin{align}  \label{U_final_alt}
U(\tau, x) &= \sum_{n=-\infty}^{\infty} \Bigg \{ \int_{y(0)}^{z(0)} U(0, \xi) \Upsilon_n(x, \tau  \,|\, \xi, 0) d\xi +\int_0^\tau \left[\Phi(s) + f^{+}(s) z'(s) \right]\Upsilon_n (x, \tau | z(s), s)ds, \nonumber \\
&\qquad +\int_0^\tau \left[\Psi(s)  - f^{-}(s) y'(s) \right] \Upsilon_n(x, \tau \,|\, y(s), s)ds \\
&\qquad + \int_0^\tau  f^-(s) \Lambda_n (x, \tau  \,|\,y(s), s) - f^+(s) \Lambda_n(x, \tau  \,|\, z(s), s) ds \Bigg\}, \nonumber \\
\Upsilon_n&(x, \tau \,|\, \xi, s) = \frac{1}{2\sqrt{\pi (\tau - s)}}\left[e^{-\frac{(2n l(\tau)  +x - \xi)^2}{4 (\tau - s)}} - e^{-\frac{(2n l(\tau)  + x +  \xi - 2 y(\tau))^2}{4 (\tau - s)}} \right],  \nonumber \\
\Lambda_n&(x, \tau \,|\, \xi, s)  = \frac{x - \xi + 2n l(\tau)}{4 \sqrt{\pi (\tau -s)^3}} e^{-\frac{(2n l(\tau)  + x  - \xi)^2}{4 (\tau - s)}} + \frac{x + \xi - 2 y(\tau) + 2n l(\tau) }{4 \sqrt{\pi (\tau -s)^3}} e^{-\frac{(2n l(\tau)  + x +  \xi - 2 y(\tau))^2}{4 (\tau - s)}}. \nonumber
\end{align}

Note that the Fourier series in these expressions usually converge rapidly when $n$ grows.  Similarly, taking the derivative of this series on $x$ provides a convenient way of calculating the corresponding derivative $\fp{U(\tau,x)}{x}$, \citep{NIST:DLMF}.

\subsection{A system of Volterra equations for $\Psi(\tau)$ and $\Phi(\tau)$}

In Section \ref{sec:Poisson} we managed to obtain two alternative representations of the solution of the problem, both in a semi-analytical form. These solutions, however, depend on two yet unknown functions gradients $\Psi(\tau), \Phi(\tau)$ that can be found by solving a system of two Volterra equations of the second kind. These equations are obtained by taking the derivative in \eqref{U_final} or \eqref{U_final_alt} with respect to $x$ and substituting  $x = y(\tau)$ and $x = z(\tau)$ into thus found expressions. However, at least formally there exist a problem with making the last step, because at these boundaries some integrals in the system of the Volterra equations will contain singularities. Below we describe the  resolution of these problems.

Let us again consider \eqref{U_final_alt}. It is easy to see that the functions $\fp{\Upsilon_n(x,\tau | \xi, s)}{x}$, $\fp{\Lambda_n(x,\tau | \xi, s)}{x}$ are regular only if $n \neq 0, \ x \in [y(\tau), z(\tau)], \ \xi \in [y(s), z(s)], \ s \to \tau$. At $n = 0$ functions $\fp{\Upsilon_0(x,\tau | y(s), s)}{x}$, $\fp{\Lambda_0(x,\tau | y(s), s)}{x}$ have a singularity when $s \to \tau,  \ x \to y(\tau)$, and functions $\fp{\Upsilon_0(x,\tau | z(s), s)}{x}, \fp{\Lambda_0(x,\tau | z(s), s)}{x}$ - when $s \to \tau,  \ x \to z(\tau)$.

Since the functions $\fp{\Upsilon_0(x,\tau | y(s), s}{x}$, $\fp{\Upsilon_0(x,\tau | z(s), s}{x}$ can be represented as a sum of double-layer potentials with a negative sign, the limiting values
\begin{equation*}
\lim_{x \to y(\tau)+0} \int_0^\tau \xi(s) \fp{\Upsilon_0(x, \tau | y(s), s)}{x} ds,
\qquad \lim_{x \to z(\tau)-0} \int_0^\tau \xi(s) \fp{\Upsilon_0(x, \tau | z(s), s)}{x} ds
\end{equation*}
\noindent can be computed similar to \eqref{Ulimits}.

Applying \eqref{finPsi2} to the limits corresponding to $\fp{\Lambda_0(x, \tau | y(s), s)}{x}, \fp{\Lambda_0(x, \tau | z(s), s)}{x}$ yields
\begin{align}
\lim_{x \to y(\tau)+0} &\int_0^\tau f^-(s) \fp{\Lambda_0(x, \tau | y(s), s)}{x} ds \\
&= -\frac{f^-(\tau)}{\sqrt{\pi \tau}} + \int_0^\tau \frac{f^-(s) e^{-\frac{(y(\tau) - y(s))^2}{4 (\tau - s)}} - f^-(\tau)}{2 \sqrt{\pi (\tau - s)^3}}ds - \int_0^\tau f^-(s) \frac{ (y(\tau) - y(s))^2 e^{-\frac{(y(\tau) - y(s))^2}{4 (\tau - s)}}}{4 \sqrt{\pi (\tau - s)^5}} ds, \nonumber \\
\lim_{x \to z(\tau)-0} &\int_0^\tau f^+(s) \fp{\Lambda_0(x, \tau | z(s), s)}{x} ds \nonumber \\
&= -\frac{f^+(\tau)}{\sqrt{\pi \tau}} + \int_0^\tau \frac{f^+(s) e^{-\frac{(z(\tau) - z(s))^2}{4 (\tau - s)}} - f^+(\tau)}{2 \sqrt{\pi (\tau - s)^3}}ds - \int_0^\tau f^+(s) \frac{ (z(\tau) - z(s))^2 e^{-\frac{(z(\tau) - z(s))^2}{4 (\tau - s)}}}{4 \sqrt{\pi (\tau - s)^5}} ds. \nonumber
\end{align}
Finally, taking the derivative of \eqref{U_final_alt} on $x$, setting $x = y(\tau)$ and $x=z(\tau)$, and using
these expressions, we obtain the following system of the Volterra equations of the second kind for the unknown functions $\Psi(\tau), \Phi(\tau)$
\begin{align} \label{Volterra_eq_alt}
-\Psi(\tau) &= -\frac{f^-(\tau)}{\sqrt{\pi \tau}} + \int_0^\tau \frac{ f^-(s) e^{-\frac{(y(\tau) - y(s))^2}{4 (\tau -s )}} \left[1 + y'(s) (y(\tau) - y(s))) -\frac{(y(\tau )- y(s))^2}{2 (\tau -s)} \right] -f^-(\tau)}{2 \sqrt{\pi (\tau - s)^3}} ds \\
&-\int_0^\tau \Psi(s) \frac{y(\tau) - y(s)}{2\sqrt{\pi (\tau -s)^3}} e^{- \frac{(y(\tau) - y(s))^2}{4(\tau -s)}} ds
 + \int_{y(0)}^{z(0)} U(0, \xi) \upsilon^-(\tau \, |\, \xi, 0) d\xi \nonumber \\
&+\int_0^\tau \left(\left[\Phi(s) + f^+(s) z'(s)  \right]\upsilon^-(\tau \,|\,z(s), s) + 	\left[\Psi(s)- f^-(s) y'(s) \right]\upsilon^-_0(\tau \, |\, y(s), s)\right)ds \nonumber \\
&+ \int_0^\tau \left( f^-(s) \lambda^-_0(\tau, \,|\, y(s), s) -f^+(s) \lambda^-(\tau, \,|\, z(s), s)  \right) ds, \nonumber \\
\Phi(\tau) &= \frac{f^+(\tau)}{\sqrt{\pi \tau}} - \int_0^\tau \frac{ 	f^+(s) e^{-\frac{(z(\tau) - z(s))^2}{4 (\tau -s )}} \left[1 +z'(s) (z(\tau) - z(s))- \frac{(z(\tau )- z(s))^2}{2 (\tau -s)} \right] -f^+(\tau)}{2 \sqrt{\pi (\tau - s)^3}} ds \nonumber \\
&-\int_0^\tau \Phi(s) \frac{z(\tau) - z(s)}{2\sqrt{\pi (\tau -s)^3}} e^{- \frac{(z(\tau) - z(s))^2}{4(\tau -s)}} ds
+ \int_{y(0)}^{z(0)} U(0, \xi) \upsilon^+(\tau \, |\, \xi, 0) d\xi \nonumber \\
&+ \int_0^\tau \left(\left[\Phi(s) + f^+(s) z'(s)  \right]\upsilon^+_{0}(\tau \,|\,z(s), s) + 	
\left[\Psi(s)- f^-(s) y'(s) \right]\upsilon^+(\tau \, |\, y(s), s)ds\right) \nonumber \\
&+ \int_0^\tau \left( f^-(s) \lambda^+(\tau \,|\, y(s), s) -f^+(s) \lambda^+_0(\tau \,|\,s)  \right) ds. 	
\nonumber
\end{align}

Here
\begin{align} \label{upsilon_lambda_def}
\upsilon_n(\tau \,|\, \xi ,s) &= - \frac{y(\tau) - \xi + 2 n l(\tau)}{2 \sqrt{\pi (\tau - s)^3}}e^{-\frac{(y(\tau) - \xi + 2 n l(\tau))^2}{4 (\tau -s)}}, \\
\lambda_n(\tau \,|\, \xi ,s) &= \frac{e^{-\frac{(y(\tau) - \xi + 2 n l(\tau))^2}{4 (\tau -s)}}}{2\sqrt{\pi (\tau - s)^3}} \left[ 1 - \frac{(y(\tau) - \xi+ 2 n l(\tau))^2}{2(\tau - s)} \right], \nonumber \\
\upsilon^-(\tau \,|\, \xi ,s) &= \sum_{n = -\infty}^{\infty} \upsilon_n(\tau \,|\, \xi ,s), \qquad
\upsilon^+(\tau \,|\, \xi ,s) = \sum_{n = -\infty}^{\infty} \upsilon_{n + \frac 12}(\tau \,|\, \xi ,s),  \nonumber \\
\upsilon^-_0(\tau \,|\, s) &= \sum_{\substack{n = -\infty\\ n \neq 0}}^{\infty}\upsilon_n(\tau \,|\, y(s) ,s), \qquad \upsilon^+_0(\tau \,|\, s) =\sum_{\substack{n = -\infty\\ n \neq 0}}^{\infty}  \upsilon_{n + \frac 12}(\tau \,|\, z(s) ,s), \nonumber\\
\lambda^-(\tau \,|\, \xi ,s) &= \sum_{n = -\infty}^{\infty} \lambda_n(\tau \,|\, \xi ,s), \qquad
\lambda^+(\tau \,|\, \xi ,s) = \sum_{n = -\infty}^{\infty} \lambda_{n + \frac 12}(\tau \,|\, \xi ,s),  \nonumber \\
\lambda^-_0(\tau \,|\, s) &= \sum_{\substack{n = -\infty\\ n \neq 0}}^{\infty} \lambda_n(\tau \,|\, y(s) ,s), \qquad \lambda^+_0(\tau \,|\, s) =\sum_{\substack{n = -\infty\\ n \neq 0}}^{\infty}  \lambda_{n + \frac 12}(\tau \,|\, z(s) ,s). \nonumber
\end{align}
It is worth emphasizing that all summands in \eqref{Volterra_eq_alt} are regular. The integrals with respect to the time in the first two lines have weak (integrable) singularities, while other summands are regular.

This system can be further simplified by using \eqref{Aux_for_Volterra} and reduction to the Lebesgue-Stieltjes integrals
\begin{align} \label{Volterra_eq_alt_LS}
-\Psi(\tau) &= \int_{y(0)}^{z(0)} U(0, \xi) \upsilon^-(\tau \, |\, \xi, 0) d\xi \\
&-\frac{f^-(\tau)}{\sqrt{\pi \tau}} + \int_0^\tau \frac{	f^-(s)-f^-(\tau)}{2 \sqrt{\pi (\tau - s)^3}} ds
+   \int_0^\tau  \left[ f^-(s)d \left(\eta^-(\tau \,|\, y(s), s)\right)  - f^+(s) d\left(\eta^-(\tau \,|\, z(s), s)\right) \right] \nonumber \\ &-\int_0^\tau \Psi(s) \frac{y(\tau) - y(s)}{2\sqrt{\pi (\tau -s)^3}} e^{- \frac{(y(\tau) - y(s))^2}{4(\tau -s)}} ds
+\int_0^\tau \left[\Phi(s) \upsilon^-(\tau \,|\,z(s), s) + \Psi(s)\upsilon^-_{0}(\tau \, |\, s)\right]ds \nonumber \\ \Phi(\tau) &= \int_{y(0)}^{z(0)} U(0, \xi) \upsilon^+(\tau \, |\, \xi, 0) d\xi \nonumber \\
&+\frac{f^+(\tau)}{\sqrt{\pi \tau}} - \int_0^\tau \frac{f^+(s) - f^+(\tau)}{2 \sqrt{\pi (\tau - s)^3}} ds
+   \int_0^\tau  \left[ f^-(s) d \left(\eta^+(\tau \,|\, y(s), s)\right)  - f^+(s) d\left(\eta^+(\tau \,|\, z(s), s)\right) \right] \nonumber \\ &-\int_0^\tau \Phi(s) \frac{z(\tau) - z(s)}{2\sqrt{\pi (\tau -s)^3}} e^{- \frac{(z(\tau) - z(s))^2}{4(\tau -s)}} ds
+\int_0^\tau \left[\Phi(s) \upsilon^+_{0}(\tau \,|\,s) + \Psi(s)\upsilon^+(\tau \, |\, y(s), s)\right] ds. \nonumber
\end{align}

Here the following notation is used
\begin{align} \label{eta_def}
\eta^-(\tau \,|\, \xi, s) &= -\frac{\delta_{\xi, y(s)}}{\sqrt{\pi (\tau -s)}}+ \frac{1}{\sqrt{\pi (\tau -s)}} \sum_{n = -\infty}^{\infty} e^{-\frac{(y(\tau) - \xi + 2n l(\tau))^2}{4(\tau -s)}}, \\
\eta^+(\tau \,|\, \xi, s) &= -\frac{\delta_{\xi,z(s)}}{\sqrt{\pi (\tau -s)}}+ \frac{1}{\sqrt{\pi (\tau -s)}} \sum_{n = -\infty}^{\infty} e^{-\frac{(y(\tau) - \xi + (2n +1) l(\tau))^2}{4(\tau -s)}}, \nonumber \\
\upsilon^-(\tau \,|\, \xi, s) &= -\frac{y(\tau) - \xi + 2 n l(\tau)}{2 \sqrt{\pi (\tau - s)^3}} e^{-\frac{(y(\tau) - \xi + 2 n l(\tau))^2}{4(\tau - s)}}, \nonumber \\
\upsilon^+(\tau \,|\, \xi, s) &= -\frac{y(\tau) - \xi + (2 n  + 1)l(\tau)}{2 \sqrt{\pi (\tau - s)^3}} e^{-\frac{(y(\tau) - \xi + (2 n + 1) l(\tau))^2}{4(\tau - s)}}, \nonumber
\end{align}
\noindent where $\delta_{\xi, x}$ is the Kronecker symbol.

The functions $\upsilon, \eta$ have the following properties
\begin{alignat*}{2}
\lim_{s \to \tau}\upsilon^-_{0}(\tau \,|\,s) &= 0 , \qquad \lim_{s \to \tau}\upsilon^-(\tau \,|\, z(s), s) &&= 0, \\
\lim_{s \to \tau}\upsilon^+(\tau \,|\, y(s), s) &= 0, \qquad \lim_{s \to \tau}\upsilon^+_{0}(\tau \,|\,s) &&= 0, \\
\lim_{s \to \tau}\eta^-(\tau \,|\, y(s), s) &= 0 , \qquad \lim_{s \to \tau}\eta^-(\tau \,|\, z(s), s) &&= 0, \\
\lim_{s \to \tau}\eta^+(\tau \,|\, y(s), s) &= 0 , \qquad \lim_{s \to \tau}\eta^+(\tau \,|\, z(s), s) &&= 0. \end{alignat*}

Again, using the Poisson summation formula yields a few alternative representations of the functions $\eta^{\pm}(\tau \,|\, \xi, s)$ and $\upsilon^{\pm}(\tau \,|\, \xi, s)$ via the Fourier series

\begin{align} \label{eta_def_Fourier}
\eta^-(\tau \,|\, \xi, s) &=-\frac{\mathbf{1}_{y(s) - \xi}}{\sqrt{\pi (\tau -s)}}
+ \frac{1}{l(\tau)} \left[ 1 + 2  \sum_{n = 1}^{\infty} e^{-\frac{\pi^2 n ^2}{l^2(\tau)} (\tau - s)} \cos \left( \frac{\pi n (\xi - y(\tau)) }{l(\tau)} \right) \right], \\
\eta^+(\tau \,|\, \xi, s) &=-\frac{\mathbf{1}_{\xi - z(s)}}{\sqrt{\pi (\tau -s)}} + \frac{1}{l(\tau)} \left[ 1 + 2  \sum_{n = 1}^{\infty} e^{-\frac{\pi^2 n ^2}{l^2(\tau)} (\tau - s)} (-1)^{n} \cos \left( \frac{\pi n (\xi - y(\tau)) }{l(\tau)} \right) \right],  \nonumber \\
\upsilon^-(\tau \,|\, \xi, s) &=\frac{2 \pi }{l^2(\tau)} \sum_{n = 1}^{\infty} n e^{-\frac{\pi^2 n ^2}{l^2(\tau)} (\tau - s)} \sin \left( \frac{\pi n (\xi - y(\tau)) }{l(\tau)} \right), \nonumber \\
\upsilon^+(\tau \,|\, \xi, s) &= \frac{2 \pi }{l^2(\tau)} \sum_{n = 1}^{\infty} n e^{-\frac{\pi^2 n ^2}{l^2(\tau)} (\tau - s)} (-1)^{n} \sin \left( \frac{\pi n (\xi - y(\tau)) }{l(\tau)} \right). \nonumber
\end{align}

Finally, using \eqref{theta_def} and \eqref{thetaRepr}, we obtain another representation of \eqref{eta_def_Fourier} in terms of the Jacobi theta function $\theta_3(z, \omega)$
\begin{align} \label{eta_def_Theta}
\eta^-(\tau \,|\, \xi, s) &=-\frac{\mathbf{1}_{y(s) - \xi}}{\sqrt{\pi (\tau -s)}}
+ \frac{1}{l(\tau)} \theta_3\left( \phi_-(\xi, y(\tau)), \omega_2 \right), \\
\eta^+(\tau \,|\, \xi, s) &=-\frac{\mathbf{1}_{\xi - z(s)}}{\sqrt{\pi (\tau -s)}} + \frac{1}{l(\tau)} \theta_3\left( \phi_-(\xi + l(\tau), y(\tau)), \omega_2 \right),  \nonumber \\
\upsilon^-(\tau \,|\, \xi, s) &=-\frac{\pi }{2 l^2(\tau)}  \theta'_3\left( \phi_-(\xi, y(\tau)), \omega_2 \right), \nonumber \\
\upsilon^+(\tau \,|\, \xi, s) &= -\frac{\pi }{2 l^2(\tau)}\theta'_3\left( \phi_-(\xi + l(\tau), y(\tau)), \omega_2 \right) . \nonumber
\end{align}

The formulas \eqref{eta_def} and \eqref{eta_def_Fourier} are complementary. Since the exponents in \eqref{eta_def_Theta} are proportional to the difference $\tau - s$, the Fourier series \eqref{eta_def_Theta} converge fast if $\tau - s$ is large. Contrary, the exponents in \eqref{eta_def} are inversely proportional to $\tau - s$. Therefore, the series \eqref{eta_def} converge fast if $\tau - s$ is small.

\section{Solution by the HP method} \label{sec:HP}

Similar to Section~\ref{sec:GIT}, the HP method, \citep{TS1963, Friedman1964, kartashov2001}, can be used to price double barrier options by solving the problem in \eqref{wEq}. The idea was first proposed and developed in \citep{ItkinMuravey2020} and is a generalization of the standard HP method for the case of two moving boundaries. Note, that to the best of authors' knowledge, yet  the closed form (or even semi-closed form) solution of this problem was not known in physics, even not mentioning finance. Below we explain our approach paying attention to all intermediate details as the behavior of the solution at the boundaries is not trivial.

Following the main idea of the HP method, let us search for the solution of the $\calB$ problem in \eqref{Heat} \eqref{bc}, \eqref{tc} in the form
\begin{equation} \label{q1}
U(\tau, x) = q(\tau,x) + \frac{1}{2\sqrt{\pi \tau}}  \int_{y(0)}^{z(0)} U(0,x')   e^{-\frac{(x-x')^2}{4\tau}} dx',
\end{equation}
\noindent so function $q(\tau,x)$ solves a problem with the homogeneous initial condition
\begin{align} \label{qDB}
\fp{q(\tau,x)}{\tau} &= \sop{q(\tau,x)}{x}, \\
q(0,x) &= 0, \qquad y(0) < x < z(0), \nonumber \\
q(\tau, y(\tau)) &= \phi_1(\tau) \equiv  f^-(\tau) - \frac{1}{2\sqrt{\pi \tau}} \int_{y(0)}^{z(0)} u(0,x')   e^{-\frac{(y(\tau)-x')^2}{4\tau}} dx', \nonumber \\
q(\tau, z(\tau)) &= \psi_1(\tau) \equiv  f^+(\tau) - \frac{1}{2\sqrt{\pi \tau}}  \int_{y(0)}^{z(0)} u(0,x')   e^{-\frac{(z(\tau)-x')^2}{4\tau}} dx'. \nonumber
\end{align}

In \citep{ItkinMuravey2020} it is proposed to search for the solution of \eqref{qDB} in the form of a generalized heat potential
\begin{equation} \label{poten1}
q(x,\tau) = \frac{1}{4\sqrt{\pi}} \int_0^\tau  \frac{1}{\sqrt{(\tau-k)^3}} \left( (x-y(k)) \Omega(k) e^{-\frac{(x-y(k))^2}{4(\tau-k)}}
+ (x-z(k))\Theta(k) e^{-\frac{(x-z(k))^2}{4(\tau-k)}} \right) dk,
\end{equation}
\noindent  where $\Omega(k), \Theta(k)$ are the heat potential densities. In other words, the solution is represented as a sum of two heat potentials: one corresponds to the lower barrier, and the other one - to the upper barrier. It is easy to check, that each such a potential solves the heat equation in \eqref{qDB}, see \citep{TS1963} as the derivative with respect to $\tau$ of the RHS of \eqref{poten1} can be pulled into the integral since the value of both integrands at $k = \tau$ vanishes.

To find the unknown functions $\Omega(k), \Theta(k)$  one can substitute into \eqref{poten1} the values $x=y(\tau)$ and $x = z(\tau)$, and get a system of two integral equations that the functions $\Omega(k), \Theta(k)$ solve. However, it is well-known, \citep{TS1963}, that these integrals at $x \to y(\tau)$  and $x \to z(\tau)$ have a discontinuity, but with the finite value at $x=y(\tau) \pm 0$ and $x=z(\tau) \pm 0$. To investigate this discontinuity in more detail and derive the value of heat potential integral at the boundary $x \to y(\tau) \pm 0$, we consider a problem similar to \eqref{qDB}
\begin{align} \label{qProblem}
\LL q(\tau,x) &= 0, \qquad  (x,\tau) \in \Omega: [y(\tau), \infty) \times \mathbb{R}_{+}, \\
q(0,x) &= 0, \quad y(0) < x < \infty, \nonumber \\
q(\tau,y(\tau)) &= \chi(\tau),  \qquad q(\tau, x)\Big|_{x \to \infty} = 0. \nonumber
\end{align}
\noindent with the operator ${\cal L}$ defined as
\begin{equation} \label{LLi_def1}
\LL =  - \frac{\partial }{\partial \tau} + \sigma^2 \frac{\partial^2}{\partial x},
\end{equation}
\noindent where $\sigma = const$. Using the HP method,  the solution of this problem can be expressed as
\begin{align} \label{udef}
q(\tau, x) &= \int_0^\tau \Omega(k) \frac{x- y(k)}{4\sigma^3 \sqrt{\pi (\tau - k)^3}} e^{-\frac{(x- y(k))^2}{4\sigma^2(\tau - k)}} d k,
\end{align}
\noindent where $\Omega(\tau)$ is the heat potential density, and $y(\tau)$ is a smooth curve (the moving boundary).
Our aim below is to derive the value of this integral at $x \to y(\tau) \pm 0$, and the gradient $\partial q(\tau,x)/\partial x$ in the same limit, namely
\begin{equation} \label{lim_defs}
\varphi(\tau) = \lim_{x \to y(\tau) \pm 0} q(\tau, x) ,\qquad \psi(\tau) = \lim_{x \to y(\tau) \pm 0} \fp{q(\tau, x)}{x}.
\end{equation}

\subsection{The limiting value of $\varphi(t)$} \label{limitPhi}

This result is obtained, eg., in \citep{TS1963}. Consider a function $W(\tau,x) = 2 \sigma^2 \phi(t)$
\begin{align} \label{Wint}
W(\tau,x)  &= \int_0^\tau \Omega(k) \frac{x - y(k)}{2\sigma \sqrt{\pi (\tau - k)^3}} e^{-\frac{(y(\tau)- y(k))^2}{4\sigma^2(\tau - k)}} dk.
\end{align}
Also consider an auxiliary integral
\begin{align} \label{Vint}
\tilde{V}(\tau,x)  &= \int_0^\tau  \frac{y'(k)  \Omega(k)}{\sigma \sqrt{\pi (\tau - k)}} e^{-\frac{(x - y(k))^2} {4\sigma^2(\tau - k)}} dk.
\end{align}
Assume that $y(k)$ is differentiable. As shown in \citep{TS1963}, $\tilde{V}(\tau,x)$ is continuous along the curve $x = y(\tau)$ because it converges uniformly and $y'(k)$ is bounded, while $W(\tau,x)$ is discontinuous. To show this, first assume that $\Omega(\tau) = \Omega_0 = const$. Then the difference $W_0 - \tilde{V}_0$, where the sub-index $_0$ means that we use $\Phi_0$ instead of $\Phi(\tau)$ in the definitions \eqref{Wint},  \eqref{Vint}, can be calculated directly with the change of variables $ k \mapsto a = (x-y(k))/(2 \sigma \sqrt{\tau - k})$
\begin{align} \label{difWV}
W_0 - \tilde{V}_0 &= \frac{1}{2 \sigma \sqrt{\pi}}\int_0^\tau  \Omega_0  e^{-\frac{(x- y(k))^2}{4\sigma^2(\tau - k)}} \left[ \frac{x - y(k)}{{(\tau - k)^{3/2}}} - \frac{2 y'(k)}{{(\tau - k)^{1/2}}} \right] dk = \Omega_0 \frac{2}{\sqrt{\pi}} \int_{\zeta^-}^{\zeta^+} e^{-a^2} da, \\
\zeta^- &= \frac{x - y(0)}{2 \sigma\sqrt{\tau}}, \qquad
\zeta^+ =
\begin{cases}
\infty, & x > y(\tau), \\
0, & x = y(\tau), \\
-\infty, & x < y(\tau).
\end{cases}
\nonumber
\end{align}
Accordingly, at, say $x \to y(\tau)+0$ we obtain
\begin{equation} \label{sqbr}
\left[W_0(\tau, y(\tau)+0)  - W_0(\tau, y(\tau)) \right] - \left[\tilde{V}_0(\tau, y(\tau)+0)  - \tilde{V}_0(\tau, y(\tau)) \right] = \Omega_0 \frac{2}{\sqrt{\pi}} \int_0^\infty e^{-a^2} da = \Omega_0.
\end{equation}
Since the function $\tilde{V}_0$ is continuous, the expression in the second square brackets in  \eqref{sqbr} vanishes, and so
\begin{equation} \label{sqbr1}
W_0(\tau, y(\tau)+0)  - W_0(\tau, y(\tau)) = \Omega_0.
\end{equation}
If $\Omega(\tau)$ is not constant, then
\begin{equation} \label{sqbr2}
W(\tau, x)  = W_0(\tau, x) - \int_0^\tau \frac{x- y(k)}{2\sigma \sqrt{\pi (\tau - k)^3}} e^{-\frac{(x-y(k))^2} {4\sigma^2(\tau - k)}} [\Omega(\tau) - \Omega(k)] dk.
\end{equation}

We assume that the boundary $y(\tau)$ and the potential density $\Omega(k)$ are differentiable functions of their arguments,  i.e., at least ${\cal C}^1$. Then the integral in \eqref{sqbr2} has the same singularity as the function $\tilde{V}(\tau,x)$, converges uniformly, and thus is a continuous function on the curve $x = y(\tau)$. This implies that
\begin{equation} \label{sqbr3}
W(\tau, x_0+0) -  W(\tau, x_0) = W_0(\tau, x_0+0) -  W_0(\tau, x_0) = \Omega(\tau),
\end{equation}
\noindent and, in particular, this is true for $x_0 = y(\tau)$.   In a similar way one can show that
\begin{equation} \label{sqbr4}
W(\tau, x_0-0)  = W_0(\tau, x_0) - \Omega(\tau),
\end{equation}

Combining these results together, we obtain the final formula for $\varphi(t)$
\begin{equation} \label{fin_varphi}
\varphi(\tau) =  \pm \frac{\Omega(\tau)}{2\sigma^2} + \int_0^\tau \Omega(k)
\frac{y(\tau)- y(k)}{4\sigma^3 \sqrt{\pi (\tau - k)^3}} e^{-\frac{(y(\tau)- y(k))^2}{4\sigma^2(\tau - k)}} dk.
\end{equation}

\subsection{The limiting value of $\psi(t)$}

Using the definition of $q(\tau,x)$ in \eqref{udef} we need an explicit formula for
\begin{align} \label{dqdx}
\psi(\tau) = \lim_{x \to y(\tau) \pm 0} \fp{q(\tau, x)}{x} &= \lim_{x \to y(\tau) \pm 0} \fp{}{x}
\int_0^\tau \Omega(k) \frac{x- y(k)}{4\sigma^3 \sqrt{\pi (\tau - k)^3}} e^{-\frac{(x- y(k))^2}{4\sigma^2(\tau - k)}} d k.
\end{align}

However, as shown in Section~\ref{limitPhi}, this integral is discontinuous at $x \to y(\tau)$ (this is an improper Riemann integral of second kind). Hence, we cannot compute $\psi(\tau)$ directly by taking derivative of $q(\tau,x)$ with respect to $x$.

Therefore, to proceed let us represent this integral as
\begin{align} \label{int12}
\int_0^\tau & \Omega(k) \frac{x- y(k)}{4\sigma^3 \sqrt{\pi (\tau - k)^3}} e^{-\frac{(x- y(k))^2}{4\sigma^2(\tau - k)}} d k = \Omega(\tau) \int_0^\tau \frac{x- y(k)}{4\sigma^3 \sqrt{\pi (\tau - k)^3}} e^{-\frac{(x- y(k))^2}{4\sigma^2(\tau - k)}} d k \\
&+ \int_0^\tau [\Omega(k) - \Omega(\tau)] \frac{x- y(k)}{4\sigma^3 \sqrt{\pi (\tau - k)^3}} e^{-\frac{(x- y(k))^2}{4\sigma^2(\tau - k)}} d k = I_1 + I_2. \nonumber
\end{align}

We showed in Section~\ref{limitPhi} that the second integral in \eqref{int12} has the same singularity as the function $\tilde{V}(\tau,x)$, converges uniformly, and thus is a continuous function on the curve $x = y(\tau)$. Then, it is a continuous function for $x \in \Re$. Thus, by the standard theorem of integral calculus we can differentiate this integral by parameter $x$, and the result is continuous in $x$, \citep{Butuzov2}
\begin{align} \label{int1}
&\ \lim_{x \to y(\tau) \pm 0} \fp{}{x} \int_0^\tau [\Omega(k) - \Omega(\tau)] \frac{x- y(k)}{4\sigma^3 \sqrt{\pi (\tau - k)^3}} e^{-\frac{(x- y(k))^2}{4\sigma^2(\tau - k)}} d k \\
&= \lim_{x \to y(\tau) \pm 0} \int_0^\tau [\Omega(k) - \Omega(\tau)] \frac{e^{-\frac{(x- y(k))^2}{4\sigma^2(\tau - k)}} }{4\sigma^3 \sqrt{\pi (\tau -k)^3}} \left( 1 - \frac{(x - y(k))^2}{2\sigma^2 (\tau - k)} \right) dk \nonumber \\
&= \int_0^\tau [\Omega(k) - \Omega(\tau)] \frac{e^{-\frac{(y(\tau)- y(k))^2}{4\sigma^2(\tau - k)}} }{4\sigma^3 \sqrt{\pi (\tau -k)^3}} \left( 1 - \frac{(y(\tau) - y(k))^2}{2\sigma^2 (\tau - k)} \right) dk. \nonumber
\end{align}

As far as the first integral $I_1$ in \eqref{int12} is concerned, it was already considered in Section~\ref{limitPhi}, and is denoted as $W_0(\tau,x) / 2\sigma^2$ in \eqref{difWV}. Since the integral on $a$ in the RHS of \eqref{difWV} can be computed explicitly, we have
\begin{equation} \label{difWV0}
W_0 - \tilde{V}_0 = \Omega_0 \frac{2}{\sqrt{\pi}} \int_{\zeta^-}^{\zeta^+} e^{-a^2} da = \Omega_0
\begin{cases}
\mathrm{Erfc} \left(\frac{x-y(0)}{2 \sigma \sqrt{\tau }}\right), & x > y(\tau), \\
- \mathrm{Erf} \left(\frac{x-y(0)}{2 \sigma \sqrt{\tau }}\right), & x = y(\tau), \\
-\mathrm{Erfc} \left(-\frac{x-y(0)}{2 \sigma \sqrt{\tau }}\right), & x < y(\tau).
\end{cases}
\end{equation}

Also, recall that the function $\tilde{V}_0 (\tau,x)$ is the continuous function along the curve $x=y(\tau)$ as $y'(\tau)$ is bounded, and the integral converges uniformly. Therefore
\begin{align} \label{der10}
\fp{W_0}{x} &= \fp{\tilde{V}_0}{x} - \Omega_0 \Lambda(\tau,x), \\
\Lambda(\tau,x) &=
\begin{cases}
\frac{1}{\sigma \sqrt{\pi \tau} } e^{-\frac{(x- y(0))^2}{4 \pi \sigma^2}}, & x > y(\tau), \\
\frac{1}{ \sigma \sqrt{\pi \tau}} e^{-\frac{(x- y(0))^2}{4 \pi \sigma^2}}, & x < y(\tau).
\end{cases}
\nonumber
\end{align}

Thus, $\Lambda(\tau,y(\tau)-0) = \Lambda(\tau,y(\tau)+0)$, hence the function $\Lambda(\tau,x)$ is differentiable at this point. This implies
\begin{equation} \label{der11}
\fp{W_0}{x} = -\Omega_0 \int_0^\tau y'(k) \frac{x- y(k)}{2\sigma^3 \sqrt{\pi (\tau - k)^3}} e^{-\frac{(x- y(k))^2}{4\sigma^2(\tau - k)}} d k -  \frac{\Omega_0}{\sigma \sqrt{\pi \tau}} e^{-\frac{(x- y(0))^2}{4 \sigma^2 \tau}}.
\end{equation}

As it was mentioned, the function $\tilde{V}_0(\tau,x)$ is continuous over the curve $x = y(\tau)$. However, its derivative with respect to $x$ at $x = y(\tau)$  in \eqref{der10} has a form of the RHS in \eqref{Wint}. Therefore, according to the result of Section~\ref{limitPhi}, in the limit $x \to y(\tau)$, again using \eqref{fin_varphi}, we obtain
\begin{align} \label{int2}
\lim_{x \to y(\tau) \pm 0} \fp{W_0}{x} = \mp \Omega_0  \frac{y'(\tau)}{\sigma^2}
-\Omega_0 \int_0^\tau y'(k) \frac{y(\tau)- y(k)}{2\sigma^3 \sqrt{\pi (\tau - k)^3}} e^{-\frac{(y(\tau) - y(k))^2}{4\sigma^2(\tau - k)}} d k - \frac{\Omega_0}{\sigma \sqrt{\pi \tau}} e^{-\frac{(y(\tau)- y(0))^2}{4 \sigma^2 \tau}}.
\end{align}

Combining \eqref{int1} and \eqref{int2} together yields the final result
\begin{align} \label{finPsi}
\psi(\tau) &= \int_0^\tau \Omega(k)  \frac{e^{-\frac{(y(\tau)- y(k))^2}{4\sigma^2(\tau - k)}} }{4\sigma^3 \sqrt{\pi (\tau -k)^3}} \left( 1 - \frac{(y(\tau) - y(k))^2}{2\sigma^2 (\tau - k)} \right) dk  - \Omega(\tau) f(\tau), \\
f(\tau) &= \pm \frac{ y'(\tau)}{2 \sigma^4}  +
\frac{1}{2\sigma^3 \sqrt{\pi \tau} } e^{-\frac{(y(\tau)- y(0))^2}{4  \sigma^2 \tau}} \nonumber \\
&+ \int_0^\tau \frac{e^{-\frac{(y(\tau) - y(k))^2}{4\sigma^2(\tau - k)}}}{4\sigma^3 \sqrt{\pi (\tau - k)^3}}
\left\{ 1 +  \frac{y'(k)[y(\tau) - y(k)]}{\sigma^2} -  \frac{(y(\tau) - y(k))^2}{2\sigma^2 (\tau - k)} \right\} d k. \nonumber
\end{align}

Thus, we proved that the derivative $\partial q(\tau,x)/\partial x$ is also discontinuous at $x=y(\tau)$, and obtained its explicit representation in \eqref{finPsi}. Note, that this derivative should not be confused with the normal (directional)  derivative of $u(\tau,x)$ which is continuous at $x=y(\tau)$.  Indeed, the function $q$, as defined in \eqref{udef}, is the double layer heat potential. The claim that this derivative is continuous is commonly referred as the Lyapunov-Tauber theorem of classic potential theory, see \citep{LyapunovBook}, and \citep{Gunter1967, Quaife2011, Costabel1990, Kristensson2009} and references therein for the extension to the multi-dimensional case.

It is worth mentioning, that the formula for $f(\tau)$ can be further simplified. Indeed
\begin{align*}
d\left(\frac{e^{-\frac{(y(\tau) - y(k))^2}{4\sigma^2(\tau - k)}}}{\sqrt{\tau - k}} \right)&= \left[
\frac{e^{-\frac{(y(\tau) - y(k))^2}{4\sigma^2(\tau - k)}}}{2\sqrt{(\tau - k)^3}}
- \frac{e^{-\frac{(y(\tau) - y(k))^2}{4\sigma^2(\tau - k)}}}{\sqrt{\tau - k}}\left( -\frac{y'(k)(y(\tau) - y(k))}{2\sigma^2 (\tau - k)} + \frac{(y(\tau) - y(k))^2}{4\sigma^2(\tau -k)^2}  \right) \right] dk \\
&= \frac{e^{-\frac{(y(\tau) - y(k))^2}{4\sigma^2(\tau - k)}}}{2\sqrt{(\tau - k)^3}} \left(1 +  \frac{y'(k)(y(\tau) - y(k))}{\sigma^2 (\tau - k)} -\frac{(y(\tau) - y(k))^2}{2\sigma^2(\tau -k)^2} \right) dk.
\end{align*}
Therefore,
\begin{equation} \label{Aux_for_Volterra}
\frac{e^{-\frac{(y(\tau) - y(k))^2}{4\sigma^2(\tau - k)}}}{2\sqrt{(\tau - k)^3}} \left(1 +  \frac{y'(k)(y(\tau) - y(k))}{\sigma^2 (\tau - k)} - \frac{(y(\tau) - y(k))^2}{2\sigma^2(\tau -k)^2} \right)dk
=d\left(\frac{e^{-\frac{(y(\tau) - y(k))^2}{4\sigma^2(\tau - k)}} - 1}{\sqrt{\tau - k}} \right) + \frac{dk}{2 \sqrt{(\tau -k)^3}}.
\end{equation}

Plugging this expression into the formula for $f(\tau)$ and integrating yields an alternative representation for $f(\tau)$
\begin{equation}
f(\tau) = \pm \frac{ y'(\tau)}{2 \sigma^4} +
\frac{1}{2\sigma^3 \sqrt{\pi \tau}}+ \int_0^\tau \frac{dk}{4\sigma^3 \sqrt{\pi (\tau - k)^3}},
\end{equation}
\noindent and for $\psi(\tau)$, respectively
\begin{align} \label{finPsi2}
\psi(\tau) &= - \Omega(\tau) \left( \frac{1}{2\sigma^3 \sqrt{\pi \tau}}
\pm \frac{ y'(\tau)}{2 \sigma^4}
\right) + \int_0^\tau   \frac{\Omega(k) e^{-\frac{(y(\tau)- y(k))^2}{4\sigma^2(\tau - k)}} - \Omega(\tau) }{4\sigma^3 \sqrt{\pi (\tau -k)^3}} dk \\
&- \int_0^\tau \Omega(k) \frac{(y(\tau) - y(k))^2 e^{-\frac{(y(\tau)- y(k))^2}{4\sigma^2(\tau - k)}} }{8\sigma^5 \sqrt{\pi (\tau -k)^5}}  dk. \nonumber
\end{align}

The last formula for the particular case $\sigma = 1/\sqrt{2}$ was also obtained in \citep{LKR2019} by using a different method.

\subsection{A system of Volterra equations}

With allowance for the representation obtained in \eqref{fin_varphi}, by substituting the limiting values $x \to y(\tau)$ and $x \to z(\tau)$ into \eqref{poten1}, we obtain a system of two integral equation for functions $\Omega(\tau), \Theta(\tau)$
\begin{align} \label{Abel2k1}
2\phi_1(\tau) &= \Omega(\tau) + \frac{1}{2\sqrt{\pi}} \int_0^{\tau}   \left( \Omega(k) \frac{y(\tau)-y(k)}{(\tau-k)^{3/2} } e^{-\frac{(y(\tau)-y(k))^2}{4(\tau-k)}}
+\Theta(k) \frac{y(\tau)-z(k)}{(\tau-k)^{3/2} } e^{-\frac{(y(\tau)-z(k))^2}{4(\tau-k)}} \right) dk, \\
2\psi_1(\tau) &= -\Theta(\tau) + \frac{1}{2\sqrt{\pi}} \int_0^{\tau} \left( \Omega(k) \frac{z(\tau)-y(k)}{(\tau-k)^{3/2} } e^{-\frac{(z(\tau)-y(k))^2}{4(\tau-k)}}
+ \Theta(k) \frac{z(\tau)-z(k)}{(\tau-k)^{3/2} }e^{-\frac{(z(\tau)-z(k))^2}{4(\tau-k)}} \right) dk. \nonumber
\end{align}
Each equation in this system is a Volterra equation of the second kind. The system can be solved, eg.,  by the Variational Iteration Method (VIM), see \citep{Wazwaz2011} with a linear complexity by using the Fast Gaussian Transform. Once this is done, the solution of our double barrier problem is found.

It is interesting that the representation of the solution gradient in \eqref{finPsi2} provides connection between the GIT and HP methods. Indeed, by definition in \eqref{PsiPhi_def} and also using \eqref{homTrans}, \eqref{q1}
\begin{align} \label{connect}
\Psi(\tau) &= - \fp{U(\tau, x)}{x} \Bigg|_{x = y(\tau)} \\
&=  - \fp{q(\tau, x)}{x} \Bigg|_{x = y(\tau)+0}
+ \frac{1}{4\sqrt{\pi \tau^3}} \int_{y(0)}^{z(0)} U(0,x') (y(\tau) -x') e^{-\frac{(y(\tau)-x')^2}{4\tau}} dx', \nonumber \\
\Phi(\tau) &= \fp{U(\tau, x)}{\tau} \Bigg|_{x = z(\tau)} \nonumber \\
&=  \fp{q(\tau, x)}{x} \Bigg|_{x = z(\tau)-0} + \frac{1}{4\sqrt{\pi \tau^3}} \int_{y(0)}^{z(0)} U(0,x') (z(\tau) -x') e^{-\frac{(z(\tau)-x')^2}{4\tau}} dx'. \nonumber
\end{align}

Therefore, once the pair $\Omega(\tau), \Theta(\tau)$ is known, the other pair $\Psi(\tau), \Phi(\tau)$ can be obtained explicitly from \eqref{connect}. The opposite is also true, i.e., once the pair $\Psi(\tau), \Phi(\tau)$ is known, the heat potential densities $\Omega(\tau), \Theta(\tau)$ can be found by solving this system of Volterra equations of the second kind. Thus, both the GIT and HP methods are interchangeable. But as was mentioned in Introduction, despite
both solutions are equal, their convergence properties are different.

\section{Discussion}

In this paper we extend the technique of semi-analytic (or semi-closed form) solutions, developed in \citep{CarrItkin2020, ItkinMuravey2020, CarrItkinMuravey2020,ItkinLiptonMuravey, LiptonKu2018, LiptonPrado2020},  to pricing double barrier options and present two approaches to solving it: the General Integral transform method and the Heat Potential method.  By semi-analytic solution we mean that first, we need to solve a system of two linear Volterra equations of the second kind, and then  the option price is represented as a one-dimensional integral.

Therefore, perhaps the main point is about efficiency and robustness of the proposed approach. As shown in \citep{CarrItkin2020, ItkinMuravey2020, CarrItkinMuravey2020,ItkinLiptonMuravey}, from the computational point of view the solution proposed by using the same technique for pricing single barrier options under various models with time-dependent barriers is very efficient and, at least theoretically, of the same complexity, or even faster than the forward finite-difference (FD) method. On the other hand, our approach provides high accuracy in computing the options prices, as this is regulated by quadrature rule used to discretize the integral kernel in \eqref{VolterraTheta}  \eqref{Volterra_eq_alt}), or in \eqref{Abel2k1}. Therefore,  the accuracy of the method in $x$ space can be easily increased by using high order quadratures.  For instance, using the Simpson instead of the trapezoid rule doesn't affect the complexity of our method but increases the accuracy, while increasing the accuracy for the FD method is not easy (i.e., it significantly increases the complexity of the method, e.g., see \citep{ItkinBook}).

As applied to pricing double barrier options - the problem considered in this paper, the difference is that instead of a single Volterra equation of the second kind we now have to solve a system of two equations, either in \eqref{VolterraTheta}  \eqref{Volterra_eq_alt}),  or in \eqref{Abel2k1}. This can be done in the same way as for the single barrier problem. The Volterra equation is solved by discretizing the kernel of the integral in time using some quadrature rule which yields a system of linear equations with respect to the discrete values of $\Psi(\tau), \Phi(\tau)$. It can be checked that the matrix of this system is of the form
\begin{align*}
\matr{M} =
\begin{pmatrix}
A & B \\
C & D \\
\end{pmatrix}
,
\end{align*}
\noindent where $A,D$ are lower triangular matrices with ones on the main diagonal, and $B,C$ are lower triangular matrices with zeros on the main diagonal. Therefore, this system can be solved by a simple Gauss elimination method (by a set of algebraic multiplications and additions) with complexity $O(2N)$ where $N$ is the number of the discretization points in $\tau$ for $\Psi(\tau), \Phi(\tau)$. Alternatively, when using \eqref{Abel2k1} or \eqref{Volterra_eq_alt}, since the kernel is proportional to Gaussians, the discrete sum approximating the integral can be computed with linear complexity $O(2N)$ using the Fast Gauss Transform, see eg., \citep{FGT2010}.

Once the vectors $\Psi(\tau), \Phi(\tau)$ (for the GIT method), or $\Omega(\tau), \Theta(\tau)$ (for the HP method)  are found, they can be substituted into \eqref{u_final2} or \eqref{U_final_alt} for the GIT method), or into  \eqref{poten1} (for the HP method). Then the final solution is obtained by computing the integral(s) numerically.
Various numerical examples illustrating this technique for a single barrier pricing problem can be found in \citep{CarrItkin2020, ItkinMuravey2020, CarrItkinMuravey2020,ItkinLiptonMuravey}. Also, those examples demonstrate that computationally our method is more efficient than both the backward and even forward FD methods (if one uses them to solve this kind of problems), while providing better accuracy and stability.

Somebody could be a bit confused of this terminology, since despite the solution is found explicitly as an integral, the latter depends of the unknown function of time $\Psi(\tau)$.  In support of this terminology, we can mention that the solution is definitely of a closed form on variable $x$. On variable $\tau$ the integrand explicitly depends on yet unknown function $\Psi(\tau)$ which solves the Volterra integral equation of the second kind. However, this equation can be solved with no iterations. Indeed, after the function $\Psi(\tau)$ is discretized on some grid in $\tau$ (so now it is represented by a finite vector $\psi$), the integral equation reduces to the linear equation for $\psi$, with the matrix being low triangular. Thus, the solution can be immediately obtained by a simple Gauss elimination with no iterations. Therefore, this is explicit and as such, the solution is given by a series of algebraic operations (substitutions). The finer is the grid, the closer is the solution to the exact one.

Also, we can make a reference to \cite{LiptonPrado2020,CarrItkinMuravey2020} where the phrase "semi-closed" was used verbatim. And in \cite{LKR2019,LiptonKu2018} it is called as "semi-analytical" solution. Going back in time, in \cite{KartashovLyubov1974, kartashov1999,kartashov2001} both GIT and HP methods are claimed as analytical. One can also look at \cite{TS1963}, page 533, subsection 2, which from the very beginning says, "Heat potentials are a convenient analytical device for solving boundary-value problems". Therefore, we think this terminology is appropriate.

Also, as mentioned in \citep{CarrItkinMuravey2020}, another advantage of the approach advocated in this paper is computation of option Greeks. Indeed, in both the HP and GIT methods the option prices  are represented in an explicit analytic form on $x$ (via the integrals on $\tau$ and the auxiliary variable $\xi$). This means that an explicit dependence of the option prices on the model parameters is available and transparent. Thus, explicit representations of the option Greeks can be obtained by a simple differentiation under the integrals. This means that the Greek values can be computed simultaneously with the option prices with almost no additional increase in the elapsed time. This is possible  because differentiation under the integrals slightly changes the integrands, while these changes could be represented as changes in weights of the quadrature scheme used to compute the integrals.

Also, the integrands in the integral representation of the solution could be treated as a product of some density function and weights. The major computational time is spent for computing the densities as they contain special functions. However, once computed the results can be saved during the calculation of prices, and then reused when computing the Greeks. Therefore, computing Greeks can be done very fast. This is also true eg., for Vega and other Greeks that cannot be computed by the FD method together with prices and require a separate run of the FD machinery. Here we don't have such a problem as differentiation of the integral representation with respect to the model parameters is done analytically.

Finally, as mentioned in \citep{ItkinMuravey2020}, the GIT and HP methods are complementary. In more detail, this means the following. Our experiments showed that performance of both the GIT and HP methods is same.  However, the GIT method produces more accurate results at high strikes and maturities (i.e. where the option price is relatively small) in contrast to the HP method which is more accurate at short maturities and low strikes. For the CIR and CEV models this behavior was explained in \citep{CarrItkinMuravey2020}, and for the Hull-White model - in \citep{ItkinMuravey2020}. Briefly, for the heat equation that we consider in this paper, the exponents in both the HP and GIT integrals are inversely proportional to $\tau$.  However, the GIT integrals contain a difference of two exponents (see the definition of $\Upsilon_n(x,\tau \,|\, \xi, s)$ in \eqref{U_final_alt} which becomes small at large $\tau$. On contrary, the HP exponent in \eqref{poten1} tends to 1 at large $\tau$. Therefore, the convergence properties of two methods are different at large $\tau$.

This situation is well known for the heat equation with constant coefficients. There exist two representation of the solution: one - obtained by using the method of images, and the other one - by the Fourier series. Despite both solutions are equal as the infinite series, their convergence properties are different.

\section*{Acknowledgments}

We are grateful to Alex Lipton for some fruitful discussions. Dmitry Muravey acknowledges support by the Russian Science Foundation under the Grant number 20-68-47030.

%%%%%%%%%%%%%%%%%%%%%%%%%%%%%%%%%%%%%%%%%%%%%%%%%%%%%%%%%%%%
%\printbibliography[title={References}]
%\bibliographystyle{plainnat}
%\bibliography{tdOU,aitkin_fin}

\vspace{0.4in}
\appendixpage
\appendix
\numberwithin{equation}{section}
\setcounter{equation}{0}

\section{Simplification of \eqref{u_final}} \label{App1}

To simplify \eqref{u_final} we proceed by integrating by parts the last integral in \eqref{u_final}
\begin{align*}
&\int_0^\tau e^{-\frac{\pi^2 n^2}{l^2(\tau)} (\tau - s)} h_1(n, s, \tau) ds =
-\frac{B(\tau) l^2(\tau)}{\pi^2 n^2} \left[\sin\left( \frac{\pi n }{l(\tau)} [z(\tau)-y(\tau)] \right) - \sin\left( \frac{\pi n }{l(\tau)}[y(\tau)-y(\tau)]\right) \right] \\
&+ \frac{B(0) l^2(\tau)}{\pi^2 k^2} e^{-\frac{\pi^2 n^2}{l^2(\tau)}\tau} \left[\sin\left( \frac{\pi n }{l(\tau)} [z(0)-y(\tau)] \right) - \sin\left( \frac{\pi n }{l(\tau)}[y(0)-y(\tau)]\right) \right] \\
&-\frac{l(\tau)}{\pi n}\Bigg[ f^-(\tau) \cos\left( \frac{\pi n }{l(\tau)} [y(\tau)-y(\tau)]\right) - f^+(\tau)\cos\left( \frac{\pi n }{l(\tau)} [z(\tau) - y(\tau)]\right)  \Bigg] \\
&+\frac{l(\tau)}{\pi n} e^{-\frac{\pi^2 n^2}{l^2(\tau)}\tau} \Bigg[ f^-(0) \cos\left( \frac{\pi n }{l(\tau)} [y(0)-y(\tau)]\right) - f^+(0)\cos\left( \frac{\pi n }{l(\tau)} [z(0) - y(\tau)]\right)  \Bigg] \\
&+ \frac{l^2(\tau)}{\pi^2 n^2} \int_0^\tau B(s) e^{-\frac{\pi^2 n^2}{l^2(\tau)}(\tau - s)}\Bigg( \frac{\pi^2 n^2}{l^2(\tau)} \left[\sin\left( \frac{\pi n }{l(\tau)} [z(s)-y(\tau)] \right) - \sin\left( \frac{\pi n }{l(\tau)}[y(s)-y(\tau)]\right) \right] \\
&+ \frac{\pi n}{l(\tau)} \left[z'(s)\cos\left( \frac{\pi n }{l(\tau)} [z(s)-y(\tau)] \right) - y'(s)\cos\left( \frac{\pi n }{l(\tau)}[y(s)-y(\tau)]\right) \right] \Bigg) ds \\
&+\frac{l(\tau)}{\pi n} \int_0^\tau f^-(s) e^{-\frac{\pi^2 n^2}{l^2(\tau)}(\tau - s)}\Bigg( \frac{\pi^2 n^2}{l^2(\tau)} \cos\left( \frac{\pi n }{l(\tau)} [y(s)-y(\tau)] \right) - \frac{\pi n}{ l(\tau)} y'(s) \sin\left( \frac{\pi n }{l(\tau)}[y(s)-y(\tau)]\right) \Bigg) ds \\
&-\frac{l(\tau)}{\pi n} \int_0^\tau f^+(s) e^{-\frac{\pi^2 n^2}{l^2(\tau)}(\tau - s)}\Bigg( \frac{\pi^2 n^2}{l^2(\tau)} \cos\left( \frac{\pi n }{l(\tau)} [y(s)-y(\tau)] \right) - \frac{\pi n}{ l(\tau)} z'(s) \sin\left( \frac{\pi n }{l(\tau)}[y(s)-y(\tau)]\right) \Bigg) ds \\
&+\frac{l(\tau)}{\pi n} \int_0^\tau B(s) e^{-\frac{\pi^2 n^2}{l^2(\tau)}(\tau - s)} \Bigg[
y'(s) \cos\left( \frac{\pi n }{l(\tau)} [y(s)-y(\tau)] \right) - z'(s) \cos\left( \frac{\pi n }{l(\tau)} [z(s)-y(\tau)] \right)
\Bigg] ds,
\end{align*}
\noindent or
\begin{align} \label{int_h1}
\int_0^\tau & e^{-\frac{\pi^2 n^2}{l^2(\tau)} (\tau - s)} h_1(n, s, \tau) ds =
\frac{l(\tau)}{\pi n}\Bigg[ (-1)^{n} f^+(\tau) -f^-(\tau) \Bigg] + \alpha(\tau, n) e^{-\frac{\pi^2 n^2}{l^2(\tau)}\tau} + \int_0^\tau e^{-\frac{\pi^2 n^2}{l^2(\tau)}(\tau-s)} \beta(\tau, s, n)ds \nonumber \\
&+ \int_0^\tau e^{-\frac{\pi^2 n^2}{l^2(\tau)}(\tau-s)} B(s) \left[\sin\left( \frac{\pi n }{l(\tau)} [z(s)-y(\tau)] \right) - \sin\left( \frac{\pi n }{l(\tau)}[y(s)-y(\tau)]\right) \right]ds,
\end{align}
\noindent where
\begin{align} \label{alpha_and_beta_def}
\alpha(\tau, n) &= \frac{B(0) l^2(\tau)}{\pi^2 n^2} \left[\sin\left( \frac{\pi n }{l(\tau)} [z(0)-y(\tau)] \right) - \sin\left( \frac{\pi n }{l(\tau)}[y(0)-y(\tau)]\right) \right] \\
& +\frac{l(\tau)}{\pi n} e^{-\frac{\pi^2 n^2}{l^2(\tau)}\tau} \Bigg[ f^-(0) \cos\left( \frac{\pi n }{l(\tau)} [y(0)-y(\tau)]\right) - f^+(0)\cos\left( \frac{\pi n }{l(\tau)} [z(0) - y(\tau)]\right)  \Bigg], \nonumber \\
\beta(\tau, s, n) &= f^-(s) \Bigg( \frac{\pi  n}{l (\tau)} \cos\left( \frac{\pi n }{l(\tau)} [y(s)-y(\tau)] \right) - y'(s) \sin\left( \frac{\pi n }{l(\tau)}[y(s)-y(\tau)]\right) \Bigg) \nonumber \\
& -f^+(s) \Bigg( \frac{\pi n}{l (\tau)} \cos\left( \frac{\pi n }{l(\tau)} [z(s)-y(\tau)] \right) - z'(s) \sin\left( \frac{\pi n }{l(\tau)}[z(s)-y(\tau)]\right) \Bigg). \nonumber
\end{align}

Now we can transform the whole term
\begin{equation} \label{wholeTerm}
\frac{2}{l(\tau)}\sum_{n= 1}^{\infty} \sin\left(\pi n  \frac{x - y(\tau)}{l(\tau)}\right) \int_0^\tau e^{-\frac{\pi^2 n^2}{l^2(\tau)} (\tau - s)} h_1(n, s, \tau) ds,
\end{equation}
\noindent which appears in \eqref{u_final}. For doing that, first let us consider the integral
\begin{equation}
\int_{y(0)}^{z(0)} u(0, \xi) \sin\left(\frac{\pi n}{l(\tau)}[\xi - y(\tau)] \right) d\xi,
\end{equation}
\noindent which is also a part of the RHS in \eqref{u_final}. Recalling that by definition in \eqref{wEq} $u(0,x) = U(0,x) - A(0) - B(0) x$, and applying another identity
\begin{align*}
\int_{y(0)}^{z(0)} \left[ A(0) + B(0) \xi \right] & \sin\left( \frac{\pi n }{l(\tau)} [\xi - y(\tau)]\right) d\xi
= \frac{l(\tau)}{\pi ^2 n^2} \Bigg\{ \pi n (A(0) + B(0) y(0)) \cos \left( \frac{\pi n (y(0) - y(\tau)}{l(\tau)}\right) \\
&- \pi n \left[A(0) + B(0) z(0)\right] \cos \left( \frac{\pi n (z(0) - y(\tau)}{l(\tau)}\right) \\
&+ B(0) l(\tau) \left[ \sin \left( \frac{\pi n (z(0) - y(\tau)}{l(\tau)}\right)  - \sin \left( \frac{\pi n (y(0) - y(\tau)}{l(\tau)}\right] \right) \Bigg\},
\end{align*}
\noindent we obtain
\begin{equation} \label{u0U}
\int_{y(0)}^{z(0)} u(0,\xi) \sin\left( \frac{\pi n }{l(\tau)} [\xi - y(\tau)]\right) d\xi = \int_{y(0)}^{z(0)} U(0,\xi) \sin\left( \frac{\pi n }{l(\tau)} [\xi - y(\tau)]\right) d\xi - \alpha(\tau, n).
\end{equation}

Therefore, the terms proportional to $\alpha(\tau,n)$ in \eqref{u_final} are cancelling out. Also, substituting  \eqref{term1} into \eqref{u_final} and moving the RHS of \eqref{term1} into the LHS of \eqref{u_final} results in the change of $u(\tau,x)$ to $U(\tau, x)$ in the LHS, and cancelling out the terms proportional to $B(s)$. Finally, introducing the new function $F(\tau, x)$
\begin{equation} \label{Fdef}
F(\tau, x) = A(\tau) + B(\tau) x - \frac{2}{\pi} \sum_{n = 1} ^{\infty} \frac{(-1)^{n- 1} f^+(\tau) + f^-(\tau)}{n} \sin\left( \frac{\pi n }{l(\tau)} [x - y(\tau)]\right)
\end{equation}
\noindent we obtain the representation of $U(\tau, x)$
\begin{align} \label{U_final2}
U(\tau, x) &= \frac{2}{l(\tau)}  \sum_{n= 1}^{\infty} \sin\left(\pi n  \frac{x - y(\tau)}{l(\tau)}\right) \Bigg\{ e^{-\frac{\pi^2 n^2}{l^2(\tau)} \tau} \int_{y(0)}^{z(0)} U(0, \xi) \sin\left(\frac{\pi n}{l(\tau)}[\xi - y(\tau)] \right) d\xi \\
&+ \int_0^\tau e^{-\frac{\pi^2 n^2}{l^2(\tau)} (\tau - s)} \Big[ \Phi(s)  \sin\left(  \frac{\pi n}{l(\tau)} [z(s) - y(\tau)]\right) + \Psi(s)  \sin\left( \frac{\pi n}{l(\tau)}[y(s) - y(\tau)]\right) \nonumber \\
&+ \beta(\tau, s, n) \Big] ds  \Bigg\} + F(\tau,x). \nonumber
\end{align}

Further, using the well-known identities, \citep{GR2007}
\begin{equation} \label{iden2}
\sum_{k =1}^{\infty} \frac{\sin kx}{k} = \frac{\pi - x}{2}, \quad 0 < x < \pi, \qquad
\sum_{k = 1}^{\infty} (-1)^{k - 1}\frac{\sin kx}{k} = \frac{x}{2}, \quad 0 < x < \pi,
\end{equation}
\noindent yields the following relationship
\begin{align} \label{term1}
\sum_{n = 1}^{\infty}&\frac{2}{\pi n}\Bigg[ (-1)^{n-1} f^+(\tau) +f^-(\tau) \Bigg] \sin\left(\pi n  \frac{x - y(\tau)}{l(\tau)}\right)
= \frac{2}{\pi} \Bigg\{ \frac{ \pi f^+(\tau)}{2} \frac{x - y(\tau)}{l(\tau)}  + \frac{f^-(\tau)}{2} \left[ \pi - \pi \frac{x - y(\tau)}{l(\tau)} \right] \Bigg\} \nonumber \\
&= \frac{f^+(\tau) - f^-(\tau)}{l(\tau)} x + \frac{f^+(\tau) y(\tau) - f^-(\tau) z(\tau) }{l(\tau)}
= -\left[A(\tau) + B(\tau) x \right], \quad x \in (y(\tau), z(\tau)).
\end{align}

With the help of \eqref{term1} we arrive at another formula for $F(\tau, x)$:
\begin{align} \label{Fdef_alt}
F(\tau, x) =
\begin{cases}
f^-(\tau), &x = y(\tau), \\
0, &x \in (y(\tau), z(\tau)), \\
f^+(\tau), &x  = z(\tau).
\end{cases}
\end{align}

Combining \eqref{U_final2} and \eqref{Fdef_alt} together, and taking into account that the Fourier series in \eqref{U_final2} is equal to zero if $x = y(\tau)$ or $x = z(\tau)$, yields
\begin{align} \label{U_final3}
U(\tau, x) =
\begin{cases}
f^-(\tau), &x = y(\tau), \\
\tilde{U}(\tau, x), &x \in (y(\tau),  z(\tau)), \\
f^+(\tau), &x = z(\tau).
\end{cases}
\end{align}

Here the function $\tilde U(\tau, x): (y(\tau), z(\tau)) \times \mathbb{R}_{+} \to \mathbb{R} $ is defined as follows
\begin{align} \label{Utilde}
\tilde{U}&(\tau, x) = \frac{2}{l(\tau)}  \sum_{n= 1}^{\infty} \sin\left(\pi n  \frac{x - y(\tau)}{l(\tau)}\right) \Bigg\{ e^{-\frac{\pi^2 n^2}{l^2(\tau)} \tau} \int_{y(0)}^{z(0)} U(0, \xi) \sin\left(\frac{\pi n}{l(\tau)}[\xi - y(\tau)] \right) d\xi \\
&+ \int_0^\tau e^{-\frac{\pi^2 n^2}{l^2(\tau)} (\tau - s)} \Big[ \Phi(s)  \sin\left(  \frac{\pi n}{l(\tau)} [z(s) - y(\tau)]\right) + \Psi(s)  \sin\left( \frac{\pi n}{l(\tau)}[y(s) - y(\tau)]\right) + \beta(\tau, s, n) \Big] ds  \Bigg\}. \nonumber
\end{align}

Note, that for the derivative $\fp{F(\tau, x)}{x}$ we have
\begin{align*}
&\fp{F(\tau, x)}{x} =B(\tau) - \frac{2}{l(\tau)} \left\{ f^+(\tau) \sum_{n = 1}^{\infty} (-1)^{n - 1} \cos\left( \frac{\pi n}{ l(\tau)} [x - y(\tau)] \right) + f^-(\tau) \sum_{n = 1}^{\infty} \cos\left( \frac{\pi n}{ l(\tau)} [x - y(\tau)]  \right) \right\} \\
&= \frac{f^+(\tau) - f^-(\tau)}{l(\tau)} - \frac{2}{l(\tau)} \left\{ f^+(\tau) \sum_{n = 1}^{\infty} (-1)^{n - 1} \cos\left( \frac{\pi n}{ l(\tau)} [x - y(\tau)] \right) + f^-(\tau) \sum_{n = 1}^{\infty} \cos\left( \frac{\pi n}{ l(\tau)} [x - y(\tau)]  \right) \right\} \\
&= \frac{2}{l(\tau)} \left\{ f^+(\tau) \left[ \frac{1}{2} +  \sum_{n = 1}^{\infty} (-1)^{n} \cos\left( \frac{\pi n}{ l(\tau)} [x - y(\tau)] \right) \right] - f^-(\tau) \left[ \frac 12 +\sum_{n = 1}^{\infty} \cos\left( \frac{\pi n}{ l(\tau)} [x - y(\tau)]  \right)  \right] \right\}.
\end{align*}
Applying well known representations for the Dirac delta function $\delta(x)$
\begin{align*}
\delta (z(\tau) - x) &= \frac{2}{l(\tau)}\left[ \frac{1}{2} +  \sum_{n = 1}^{\infty} (-1)^{n} \cos\left( \frac{\pi n}{ l(\tau)} [x - y(\tau)] \right) \right], \\
\delta (x - y(\tau))  &= \frac{2}{l(\tau)}\left[ \frac 12 +\sum_{n = 1}^{\infty} \cos\left( \frac{\pi n}{ l(\tau)} [x - y(\tau)]  \right)  \right]
\end{align*}
\noindent yields the following formula for the derivative of $F(\tau, x)$
\begin{align} \label{eq_dFdx}
\fp{F(\tau, x)}{x} = f^+(\tau) \delta(x -z(\tau))- f^-(\tau) \delta(x - y(\tau)).
\end{align}
Thus, this derivative is defined only in the sense of distributions.

\section{Transformation of \eqref{U_final} to \eqref{U_final_alt}}.  \label{App2}

Applying a product-to-sum trigonometric identities to \eqref{Utilde} yields
\begin{align} \label{Utilde2}
\tilde{U}(&\tau, x) = \frac{1}{l(\tau)}  \sum_{n= 1}^{\infty} \Bigg\{  e^{-\frac{\pi^2 n^2}{l^2(\tau)} \tau} \int_{y(0)}^{z(0)} U(0, \xi) \left[ \cos\left(\frac{\pi n}{l(\tau)}[x - \xi] \right) - \cos\left(\frac{\pi n}{l(\tau)}[x + \xi - 2 y(\tau)] \right) \right] d\xi \\
&+ \int_0^\tau e^{-\frac{\pi^2 n^2}{l^2(\tau)} (\tau - s)} \left[\Phi(s) + f^+(s) z'(s)\right] \left[ \cos\left(\frac{\pi n}{l(\tau)}[x - z(s)] \right) - \cos\left(\frac{\pi n}{l(\tau)}[x + z(s) - 2 y(\tau)] \right) \right] ds \nonumber \\
&+ \int_0^\tau  e^{-\frac{\pi^2 n^2}{l^2(\tau)} (\tau - s)} \left[\Psi(s) - f^-(s) y'(s) \right]\left[ \cos\left(\frac{\pi n}{l(\tau)}[x - y(s)] \right) - \cos\left(\frac{\pi n}{l(\tau)}[x + y(s) - 2 y(\tau)] \right) \right]  ds
\nonumber \\
&+\frac{\pi n}{l(\tau)} \int_0^\tau  e^{-\frac{\pi^2 n^2}{l^2(\tau)} (\tau - s)} f^-(s) \left[
\sin\left(\frac{\pi n}{l(\tau)}[x - y(s)] \right) +\sin\left(\frac{\pi n}{l(\tau)}[x + y(s) - 2 y(\tau)] \right) \right]ds
\nonumber \\
&-\frac{\pi n}{l(\tau)} \int_0^\tau  e^{-\frac{\pi^2 n^2}{l^2(\tau)} (\tau - s)} f^+(s) \left[
\sin\left(\frac{\pi n}{l(\tau)}[x - z(s)] \right) +\sin\left(\frac{\pi n}{l(\tau)}[x + z(s) - 2 y(\tau)] \right) \right]ds
\Bigg\}, \nonumber
\end{align}

Since the functions
\begin{equation*}
h_1(n) =  e^{-\beta n^2} \cos\left(\alpha n\right) ,\qquad  h_2(n) = n e^{-\beta n^2} \sin\left(\alpha n\right)
\end{equation*}
\noindent are even, $h_2(0) = 0$, and in the first three lines of \eqref{Utilde2} we have a difference of cosines, so at $n=0$ the difference vanishes, the series in \eqref{Utilde2} can be slightly modified by replacing
\begin{equation*}
\sum_{n= 1}^{\infty} h_i(n) = \frac{1}{2}  \sum_{n= -\infty}^{\infty} h_i(n), \quad i=1,2.
\end{equation*}

Now applying formulas \eqref{Summation_formula_main_eq} to \eqref{Utilde2} and using
\begin{equation*}
\alpha = \frac{x - \xi }{l(\tau)}, \quad \frac{1}{2\beta} = \frac{\tau}{l^2(\tau)}, \quad  2 \sqrt{\frac{\beta}{2\pi}} =  \frac{l(\tau)}{\sqrt{\pi \tau}}, \quad \frac \beta 2 (2 n + \alpha)^2 = \frac{l^2(\tau)}{4\tau} \left(2n + \frac{x - \xi }{l(\tau)} \right)^2 = \frac{(x - \xi + 2n l(\tau))^2}{4\tau},
\end{equation*}
\noindent we obtain the following identities
\begin{align} \label{sums_ref}
\frac{1}{2 l(\tau)}  \sum_{n= -\infty}^{\infty} e^{-\frac{\pi^2 n^2}{l^2(\tau)} (\tau - s)}
\cos\left(\frac{\pi n}{l(\tau)}[x - \xi] \right)
&= \frac{1}{2 \sqrt{\pi (\tau - s)}} \sum_{n = -\infty} ^{\infty} e^{-\frac{(x - \xi + 2n l(\tau))^2}{4(\tau - s)}} \\
\frac{1}{2 l(\tau)}  \sum_{n= -\infty}^{\infty} e^{-\frac{\pi^2 n^2}{l^2(\tau)} (\tau - s)}
\frac{\pi n}{l(\tau)}  \sin\left(\frac{\pi n}{l(\tau)}[x - \xi] \right)
&= \frac{1}{4 \sqrt{\pi (\tau - s)^3}} \sum_{n = -\infty} ^{\infty} (x - \xi + 2 n l(\tau)) e^{-\frac{(x - \xi + 2n l(\tau))^2}{4(\tau - s)}}.
\nonumber
\end{align}

Observe that each term in \eqref{Utilde2} can be represented as one of the series in \eqref{sums_ref}. Therefore, assuming $x \in (y(\tau),z(\tau))$, we immediately arrive at the alternative representation for $\tilde U(\tau, x)$
\begin{align} \label{sol_u}
\tilde{U}(\tau, x) &= \sum_{n=-\infty}^{\infty} \Bigg \{ \int_{y(0)}^{z(0)} U(0, \xi) \Upsilon_n(x, \tau  \,|\, \xi, 0) d\xi +\int_0^\tau \left[\Phi(s) + f^{+}(s) z'(s) \right]\Upsilon_n (x, \tau | z(s), s)ds, \nonumber \\
&+\int_0^\tau \left[\Psi(s)  - f^{-}(s) y'(s) \right] \Upsilon_n(x, \tau \,|\, y(s), s)ds \\
&+ \int_0^\tau  f^-(s) \Lambda_n (x, \tau  \,|\,y(s), s) - f^+(s) \Lambda_n(x, \tau  \,|\, z(s), s) ds \Bigg\} + F_1(\tau,x), \nonumber \\
\Upsilon_n&(x, \tau \,|\, \xi, s) = \frac{1}{2\sqrt{\pi (\tau - s)}}\left[e^{-\frac{(2n l(\tau)  +x - \xi)^2}{4 (\tau - s)}} - e^{-\frac{(2n l(\tau)  + x +  \xi - 2 y(\tau))^2}{4 (\tau - s)}} \right],  \nonumber \\
\Lambda_n&(x, \tau \,|\, \xi, s)  = \frac{x - \xi + 2n l(\tau)}{4 \sqrt{\pi (\tau -s)^3}} e^{-\frac{(2n l(\tau)  + x  - \xi)^2}{4 (\tau - s)}} + \frac{x + \xi - 2 y(\tau) + 2n l(\tau) }{4 \sqrt{\pi (\tau -s)^3}} e^{-\frac{(2n l(\tau)  + x +  \xi - 2 y(\tau))^2}{4 (\tau - s)}}. \nonumber
\end{align}

\subsection{The limiting values $x \to y(\tau)$ and $x \to z(\tau)$ in \eqref{sol_u}} \label{AppLimits}

The \eqref{sol_u} provides an alternative representation of the solution $\tilde{U}(\tau,x)$ of the heat equation in \eqref{Heat}  with the initial condition in \eqref{tc} and the boundary conditions in \eqref{bc} at the time-dependent domain $x \in (y(\tau),z(\tau))$ in terms of the Fourier series. In this section we show that the function $\tilde{U}$ can be analytically continued to the boundary points $y(\tau)$ and $z(\tau)$, and
\begin{equation} \label{limU}
\lim_{x \to y(\tau) + 0} \tilde U(\tau, x) = f^-(\tau), \qquad \lim_{x \to z(\tau) - 0} \tilde U(\tau, x) = f^+(\tau).
\end{equation}
It is easy to check that the functions $\Upsilon_n(x,\tau | \xi, s)$ and $\Lambda_n(x,\tau | \xi, s)$ are regular only if $n \neq 0, \ x \in [y(\tau), z(\tau)], \ \xi \in [y(s), z(s)], \ s \to \tau$. In this case the following identities hold
\begin{align}
\lim_{s \to \tau}\Upsilon_n(x,\tau | \xi, s) = 0, \qquad
\lim_{s \to \tau}\Lambda_n(x,\tau | \xi, s) = 0, \qquad  n \neq  0.
\end{align}

At $n = 0$ functions $\Upsilon_0(x,\tau | y(s), s)$ and  $\Lambda_0(x,\tau | y(s), s)$ have a singularity when $s \to \tau,  \ x \to y(\tau)$, and functions $\Upsilon_0(x,\tau | z(s), s)$ and $\Lambda_0(x,\tau | z(s), s)$ - when $s \to \tau,  \ x \to z(\tau)$. Note, that the singularity of $\Upsilon_0$ is integrable and so weak. Therefore, when calculating a corresponding limit of both parts in \eqref{sol_u}, for the regular terms we can switch the order of the integration and limit operators, and then use the following properties
\begin{alignat}{2} \label{UpsilonLambdaLimits}
\lim_{x \to y(\tau) + 0} \sum_{n = -\infty}^{\infty} \Upsilon_n(x, \tau \,|\,\xi ,s)  &= 0, \qquad
\lim_{x \to z(\tau) - 0}  \sum_{n = -\infty}^{\infty} \Upsilon_n(x, \tau \,|\,\xi ,s)  &&= 0, \\ \nonumber
\lim_{x \to y(\tau) + 0} \sum_{\substack{n = -\infty \\ n \neq 0}} \Lambda_n(x, \tau \,|\,\xi ,s)  &= 0, \qquad
\lim_{x \to z(\tau) - 0} \sum_{\substack{n = -\infty \\ n \neq 0}} \Lambda_n(x, \tau \,|\,\xi ,s)  &&= 0, \nonumber \\
\lim_{x \to y(\tau) + 0} \Lambda_0(x, \tau \,|\,z(s) ,s)  &= 0, \qquad
\lim_{x \to z(\tau) - 0} \Lambda_0(x, \tau \,|\,y(s) ,s)  &&= 0, \nonumber
\end{alignat}
\noindent to obtain
\begin{align} \label{limLambda}
\lim_{x \to y(\tau) + 0} \tilde U(\tau, x) &= \lim_{x \to y(\tau) + 0}\int_0^\tau  f^-(s) \Lambda_0 (x, \tau  \,|\,y(s), s) \\
\lim_{x \to z(\tau) - 0} \tilde U(\tau, x) &= -\lim_{x \to z(\tau) - 0}\int_0^\tau  f^+(s) \Lambda_0 (x, \tau  \,|\,z(s), s). \nonumber
\end{align}

To proceed we need the notion of heat potentials and the results obtained in Section~\ref{sec:HP} (see also \citep{TS1963}). It can be shown that the functions $\Lambda_0(x,\tau | y(s), s)$ and $ \Lambda_0(x,\tau | z(s), s)$ can be represented as a sum of double layer heat potentials. Therefore, we can evaluate the limits \eqref{limLambda} with the help of \eqref{lim_defs}, \eqref{fin_varphi}.

In more detail, according to \eqref{limLambda} in the explicit form the limits of $\tilde{U}(\tau,x)$ read
\begin{align*}
\tilde U(\tau,y(\tau)) = \lim_{x \to y(\tau)+0} \int_0^\tau f^{-}(s) \left[ \frac{x - y(s)}{4\sqrt{\pi (\tau - s)^3}} e^{-\frac{(x - y(s))^2}{4(\tau - s)}} +  \frac{x - 2y(\tau) + y(s)}{4\sqrt{\pi (\tau - s)^3}} e^{-\frac{(x -2 y(\tau) +  y(s))^2}{4(\tau - s)}} \right] ds, \\
\tilde U(\tau,z(\tau)) = -\lim_{x \to z(\tau)-0} \int_0^\tau f^{+}(s) \left[\frac{x - z(s)}{4\sqrt{\pi (\tau - s)^3}} e^{-\frac{(x - z(s))^2}{4(\tau - s)}} + \frac{x - 2z(\tau) + z(s)}{4\sqrt{\pi (\tau - s)^3}} e^{-\frac{(x -2 z(\tau) +  z(s))^2}{4(\tau - s)}} \right] ds.
\end{align*}

This can also be re-written in the form
\begin{align} \label{Ulimits}
\tilde U(\tau,y(\tau)) &= \lim_{x \to y(\tau)+0} \int_0^\tau f^{-}(s) \frac{x - y(s)}{4\sqrt{\pi (\tau - s)^3}} e^{-\frac{(x - y(s))^2}{4(\tau - s)}} ds \\
-&  \lim_{2 y(\tau) - x \to y(\tau)-0} \int_0^\tau f^{-}(s) \frac{2y(\tau) -x  - y(s)}{4\sqrt{\pi (\tau - s)^3}} e^{-\frac{(x -2 y(\tau) +  y(s))^2}{4(\tau - s)}}  ds, \nonumber \\
\tilde U(\tau,z(\tau)) &= -\lim_{x \to z(\tau)-0} \int_0^\tau f^{+}(s) \frac{x - z(s)}{4\sqrt{\pi (\tau - s)^3}} e^{-\frac{(x - z(s))^2}{4(\tau - s)}} ds \nonumber \\
+&  \lim_{2 z(\tau) - x \to z(\tau) + 0} \int_0^\tau f^{+}(s) \frac{2z(\tau) -x  - z(s)}{4\sqrt{\pi (\tau - s)^3}} e^{-\frac{(x -2 z(\tau) +  z(s))^2}{4(\tau - s)}}  ds. \nonumber
\end{align}

Using \eqref{fin_varphi}, these expressions can be transformed to
\begin{equation*}
\tilde U(\tau, y(\tau)) = \frac{f^-(\tau)}{2} + \frac{f^-(\tau)}{2} = f^-(\tau), \qquad
\tilde U(\tau, z(\tau)) = \frac{f^+(\tau)}{2} + \frac{f^+(\tau)}{2} = f^+(\tau).
\end{equation*}

Since $\tilde{U}(\tau,x)$ has the same limits as the boundary values of $U(\tau,x)$, and at $x \in (y(\tau),z(\tau))$ we have $\tilde{U}(\tau,x) = U(\tau,x)$, \eqref{sol_u} allows an alternative form
\begin{align*}
U(\tau, x) &= \sum_{n=-\infty}^{\infty} \Bigg \{ \int_{y(0)}^{z(0)} U(0, \xi) \Upsilon_n(x, \tau  \,|\, \xi, 0) d\xi +\int_0^\tau \left[\Phi(s) + f^{+}(s) z'(s) \right]\Upsilon_n (x, \tau | z(s), s)ds, \nonumber \\
&\qquad +\int_0^\tau \left[\Psi(s)  - f^{-}(s) y'(s) \right] \Upsilon_n(x, \tau \,|\, y(s), s)ds \\
&\qquad + \int_0^\tau  f^-(s) \Lambda_n (x, \tau  \,|\,y(s), s) - f^+(s) \Lambda_n(x, \tau  \,|\, z(s), s) ds \Bigg\}. \nonumber
\end{align*}

\end{document}